\documentclass[a4paper, onecolumn, 12pt, unpublished]{quantumarticle}

\pdfoutput=1
\usepackage[utf8]{inputenc}
\usepackage[english]{babel}
\usepackage[T1]{fontenc}
\usepackage{amsmath}
\usepackage{hyperref}
\usepackage{cleveref}
\usepackage{dsfont}
\usepackage{amsfonts}
\usepackage{amssymb}
\usepackage{amsthm}
\usepackage{csquotes}
\usepackage{mathrsfs}
\usepackage{geometry}
\usepackage[sorting=none]{biblatex}
\bibliography{references}

\usepackage{lipsum}
\newtheorem{theorem}{Theorem}
\newtheorem{lemma}[theorem]{Lemma}
\newtheorem{corollary}[theorem]{Corollary}
\newtheorem{definition}[theorem]{Definition}
\theoremstyle{definition}

\newtheorem*{remark}{Remark}

\DeclareMathOperator*{\argmax}{arg\,max}
\DeclareMathOperator*{\argmin}{arg\,min}
\newcommand{\Chi}{\mathrm{X}}

\begin{document}

\title{Quantum mean states are nicer than you think: fast algorithms to compute states maximizing average fidelity}

\author{A. Afham}
\affiliation{Centre for Quantum Software and Information, University of Technology Sydney, 
NSW 2007, Australia} 
\affiliation{Sydney Quantum Academy, Sydney, NSW 2000, Australia}
\email{afham@student.uts.edu.au}

\author{Richard Kueng}
\affiliation{Institute for Integrated Circuits, Johannes Kepler University Linz, 
4040 Linz, Austria} 

\author{Chris Ferrie}
\affiliation{Centre for Quantum Software and Information, University of Technology Sydney, 
NSW 2007, Australia} 

\begin{abstract}
Fidelity is arguably the most popular figure of merit in quantum sciences. However, many of its properties are still unknown. In this work, we resolve the open problem of maximizing average fidelity over arbitrary finite ensembles of quantum states and derive new upper bounds. We first construct a semidefinite program whose optimal value is the maximum average fidelity and then derive fixed-point algorithms that converge to the optimal state. The fixed-point algorithms outperform the semidefinite program in terms of numerical runtime. We also derive expressions for near-optimal states that are easier to compute and upper and lower bounds for maximum average fidelity that are exact when all the states in the ensemble commute. Finally, we discuss how our results solve some open problems in Bayesian quantum tomography.

\end{abstract}

\maketitle

\section{Introduction}

Comparing quantum states is not only of great practical importance but also provides novel mathematical challenges since quantum states form a constrained set of complex-valued matrices. While some methods of comparison are now ubiquitous, still much is unknown about the quantities of interest derived from them. In this paper, we consider the most commonly used comparator --- namely, \textit{fidelity} --- and demonstrate both analytic and numerical algorithms for producing states which maximize averages of it over arbitrary finite ensembles. These algorithms can be straightforwardly applied, for example, in tomography, where optimal solutions had previously been lacking.

Concretely, suppose we are given a collection of quantum states $\{\rho_1, \ldots, \rho_n\}$ and tasked with the problem of producing the \textit{best} state to represent this ensemble. An intuitive notion of ``best'' is the state which is as close as possible to each state. The obvious tension is resolved by minimizing the \textit{average} distance to each state in the collection. Though it is not formally a distance in the mathematical sense, the most commonly used measure of ``closeness'' for quantum states is fidelity. The fidelity between two quantum states $\rho$ and $\sigma$  is defined as \cite{NielsenChuang}, \begin{equation}
\operatorname{F}(\rho,\sigma)  =  \operatorname{Tr} \left(\sqrt{\sigma^{1/2}\rho \sigma^{1/2}}\right).    
\end{equation}

\begin{figure*} 
    \center
    \includegraphics[scale = 0.635]{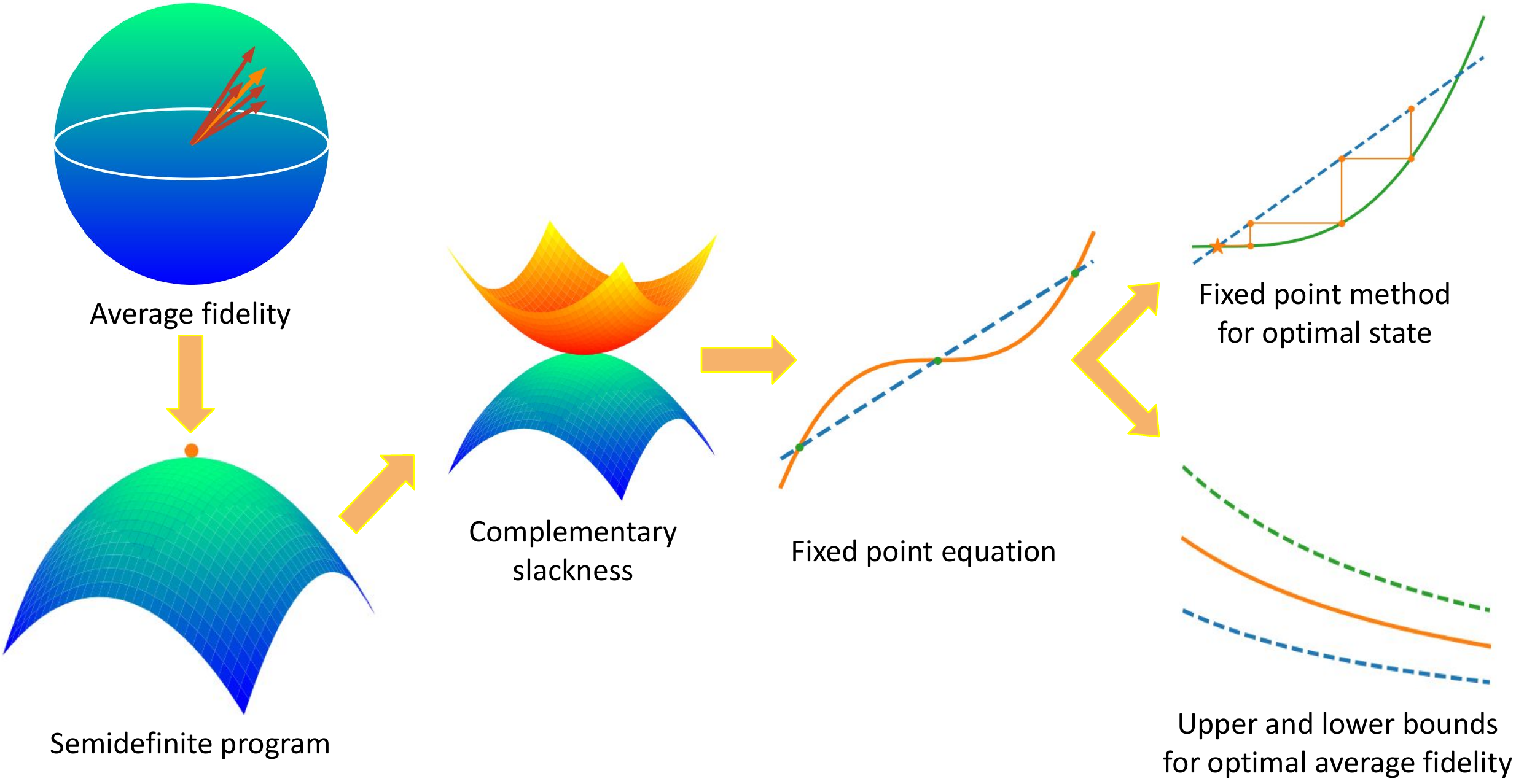}
    \caption{The problem of finding the state that maximizes the average fidelity is framed as a semidefinite program that exhibits complementary slackness relations. These relations lead to a fixed point equation satisfied by the optimal state from which we construct a fixed point iteration algorithm for the optimal state and heuristic near-optimal estimator which is optimal when all the states commute. Finally, we present upper bounds for optimal average fidelity achieved by any state.} \label{Fig:Flowchart}
\end{figure*}

The fidelity $\operatorname{F}(\rho,\sigma)$ between two quantum states can take on any value between $0$ and $1$, with $\operatorname{F}(\rho,\sigma) = 1$ if and only if $\rho = \sigma$ and $0$ when the states $\rho$ and $\sigma$ are orthogonal. More generally, each state $\rho_j$ in the collection can be associated with some probability (weight) $p_j$. Then, for any state $\sigma$, the \textit{average fidelity} is defined as,
\begin{equation}
    f(\sigma) = \sum_{j=1}^{n} p_j \operatorname{F}(\rho_j,\sigma).     
\end{equation}
Throughout this paper, we are concerned with \textit{maximizing} average fidelity through the following optimization problem,
\begin{equation} \label{eq:maximprob}
\begin{aligned}
\text{maximize : }& f(\sigma) = \sum_{j=1}^{n} p_j \operatorname{F}(\rho_j,\sigma),\\
\text{subject to : } 
&\sigma \geq 0,\; \operatorname{Tr}(\sigma) = 1. \\
\end{aligned}
\end{equation}
This is a well-posed convex optimization problem as we are maximizing average fidelity, a concave function, over the convex set of quantum states. The state $\sigma_\sharp$ which solves this optimization problem is referred to as the \textit{optimal state}.

This problem arises naturally in quantum state tomography, which is the task of estimating a quantum state that has produced a given set of measurement data. There are many frameworks for addressing the tomography problem and each has within them many procedures which generate estimates given data. Each is referred to as an \textit{estimator}, with the Maximum Likelihood Estimator (or MLE)~\cite{Hradil97} being a canonical example.

In evaluating a potential estimator, it is commonplace to run real or simulated experiments over a variety of state preparation procedures. Such randomization is desired to ensure unbiasedness, for example. The results of these experiments are the reported fidelities achieved by the estimator(s) \textit{averaged} over the ensemble chosen in the experiment. Crucially, there exists an \textit{optimal} estimator for this ensemble, where the probabilities are defined by the likelihood function (or Born rule). The average fidelity of the optimal estimator would provide an absolute benchmark in the evaluation of an estimator. Until now, no recipe to produce such a state has been provided.

There is a duality in the above discussion to Bayesian quantum tomography \cite{Blume2010BME}, where the optimal quantum state is that with maximum fidelity averaged over a posterior distribution. This state is typically called the \textit{Bayes estimator}. Hence, our work also solves this open Bayesian tomography problem. In particular, we provide iterative fixed point algorithms (Eq.~\eqref{eq:FPMaps}) which converge to the optimal state for any starting point. We provide an easy-to-compute approximation for the optimal state (Eq.~\eqref{eq:Approximator}) which is exact when all the states in the ensemble commute. We also present expressions for upper bounds on average fidelity (Eq.~\eqref{eq:informal_bounds}) that are of practical and theoretical interest. We complement our theoretical findings with numerical experiments comparing the performance and quality of our results in practical scenarios.

Related work in this direction provides bounds for optimal square fidelity~\cite{Kueng2015}, optimal square fidelity estimators for qubit states~\cite{Bagan2006}, Bayes estimators for fidelity and square fidelity restricted to diagonal density matrices~\cite{Ferrie2016}, minimization of average Bures distance over a finite ensemble of positive definite matrices~\cite{Bhatia2019}, and minimizing Rényi entropy via Mirror descent~\cite{You2021}. Our work improves upon or complements these and we hope our results can serve as a model for finding Bayes estimators of other figures of merit such as trace distance, square fidelity and so on.

The remainder paper is structured as follows. In Section~\ref{sec:Informal}, we outline the main results. In Section~\ref{sec:Prelim} we discuss the preliminaries and notations required for the rest of the article. In Section~\ref{sec:Results} we formally state the results and prove them. In Section~\ref{sec:applications} we briefly discuss the application of our theoretical results to the tomography problem and verify their utility through simulated experiments. In Section~\ref{Sec:Numerics} we provide numerical simulations relevant to our results, and we conclude in Section~\ref{Sec:Conclusion}.

\begin{figure*}
    \center
    \includegraphics[scale = 0.65]{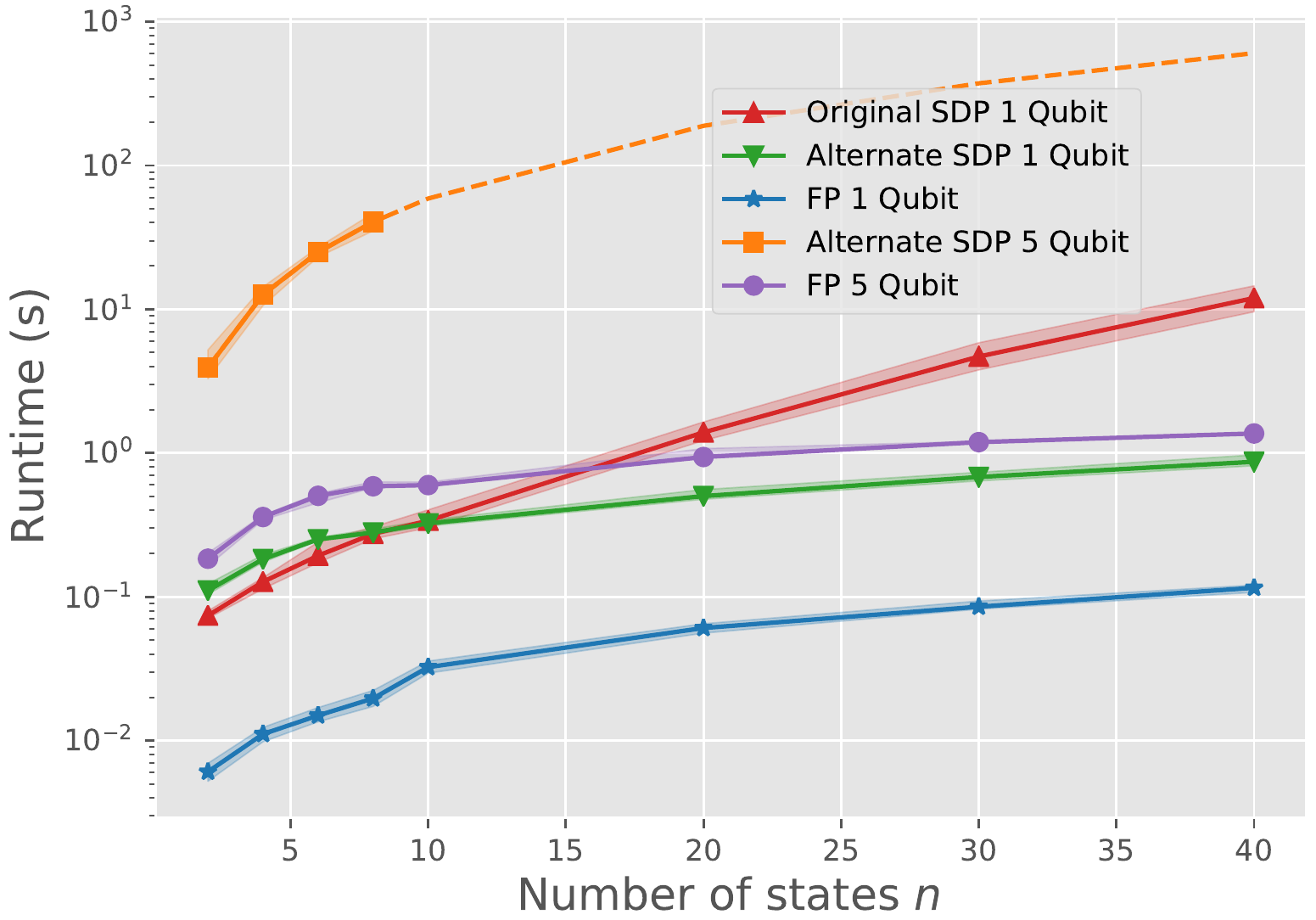}
    \caption{Performance comparison of SDP and FP algorithm for finding the state maximizing average fidelity. runtime plotted as a function of number of states for original SDP~\eqref{def:SDP}, alternate SDP~\eqref{eq:AltSDP}, and fixed point iteration algorithm~\eqref{eq:Omega} for two different dimensions. For 1 qubit states ($d = 2$), we plot the runtime of all three methods. For 5 qubit states ($d = 32$), we drop the original SDP and plot the alternate SDP only upto $n=8$ due to intractability. Each marker is the median of 50 iterations and shaded regions correspond to interquarile regions. Dashed line (for alternate SDP 5 qubit) correspond to extrapolation from numerical data. } \label{Fig:SDPvFP}
\end{figure*}

\section{Informal statement of results} \label{sec:Informal}

We first state the problem. Let $\mathcal H = \mathbb{C}^d$ be a $d$-dimensional Hilbert space and $\operatorname{D}(\mathcal H)$ denote the collection of all quantum states associated with the space $\mathcal H$:
\begin{equation}
    \operatorname{D}(\mathcal H) = \{\rho: \rho \in \operatorname{Pos}(\mathcal H), \operatorname{Tr}(\rho) = 1\},
\end{equation}
where $ \operatorname{Pos}(\mathcal H) $ is the collection of positive semidefinite operators in the space $\mathcal H$. We denote the standard simplex in $n$ dimensions, the collection of all $n$-dimensional vectors with non-negative entries summing up to unity, by $\Delta_n$:
\begin{equation}
    \Delta_n = \left\{p \in \mathbb{R}^n : p_i \geq 0, \sum_{i=1}^{n} p_i = 1 \right\}. 
\end{equation}
We call an element $p \in \Delta_n$ of the standard simplex a \textit{probability vector}. Given a collection of quantum states $\mathsf R = \{\rho_1, \ldots, \rho_n\} \subset \operatorname{D}(\mathcal H)$  and a probability vector $p \in \Delta_n$ over it, we define the average fidelity of any state $\sigma \in \operatorname{D}(\mathcal H)$  over the ensemble $(\mathsf R,p)$ as

\begin{equation}
    f(\sigma) = \sum_{i=1}^{n} p_i \operatorname{F}(\rho_i,\sigma).
\end{equation}
The optimization problem of interest is then
\begin{equation} \label{prob:maxavfid}
    \argmax_{\sigma \in \operatorname{D}(\mathcal H)}f(\sigma) = \argmax_{\sigma \in \operatorname{D}(\mathcal H)} \sum_{i=1}^{n} p_i\operatorname{F}(\rho_i,\sigma).
\end{equation}
We first present a semidefinite program (SDP) which solves Problem \eqref{prob:maxavfid} in Section~\ref{sec:OptFidSDP}. The SDP exhibits strong duality and, when all the states in the ensemble are full rank, complementary slackness. However, numerically solving SDPs can quickly grow intractable, especially since the SDP involves optimizing over matrices of dimension $(n+1)d$. This difficulty can be circumvented by the following observation. Complementary slackness of the SDP implies a fixed point equation that is satisfied by the optimal state:
\begin{equation}
    \sigma_\sharp = \frac1{f(\sigma_\sharp)} \sum_{i=1}^{n} p_i \sqrt{\sigma_\sharp^{1/2} \rho_i \sigma_\sharp^{1/2}}. 
\end{equation}
We use this to develop two fixed point iteration algorithms which converge to the optimal state for \emph{any} starting point. The two algorithms are defined by the two fixed point iteration maps $\Lambda$ and $\Omega$ of the form
\begin{equation} \label{eq:FPMaps}
\begin{aligned}
    \Lambda(\sigma) &= \Gamma \left(\sum_{i=1}^{n} p_i \sqrt{\sigma^{1/2} \rho_i \sigma^{1/2}}
 \right), \\
\Omega(\sigma) &= \Gamma \left( \sigma^{-1/2}  \left(
\sum_{i=1}^{n} p_i \sqrt{\sigma^{1/2} \rho_i \sigma^{1/2}}
 \right)^2 \sigma^{-1/2}  \right),
\end{aligned}    
\end{equation}
where $\Gamma(A) = A/\operatorname{Tr}(A)$ is used to normalise these positive definite matrices to density matrices. The sequence of states $\{\Lambda^k(\sigma)\}_{k=0}^\infty$ and $\{\Omega^k(\sigma)\}_{k=0}^\infty$ converges to the optimal state $\sigma_\sharp$ for any full rank initial state $\sigma \in \operatorname{D}(\mathcal H)$. Here we define $\Lambda^0 (\sigma) = \sigma$ and $\Lambda^k(\sigma) = \Lambda\left(\Lambda^{k-1}(\sigma)\right)$ for all integers $k \geq 1$. Similar notation is followed for the map $\Omega$.
We note that the fixed point method is guaranteed to work when all the states in the ensemble are full-rank. If one is interested in optimizing over rank-deficient states, first depolarize the states by a small factor to obtain full-rank states and then use the fixed point algorithms. Alternatively, we may simply use the SDP which yields the optimum even when the states are rank-deficient.

When all the states in $\mathsf R = \{\rho_i\}_{i=1}^n$ pairwise commute, then there exists a simple analytic expression for the optimal state $\sigma_\sharp$:
\begin{equation} \label{eq:Approximator}
    \sigma_\sharp = \sigma' = \Gamma \left( \left( \sum_{i=1}^n p_i \rho_i^{1/2}\right)^2\right). 
\end{equation}
The state $\sigma'$, called the \textit{Commuting estimator}, can also serve as an easy-to-compute near-optimal heuristic approximation even in cases where the ensemble does not commute. Numerically we see this approximation to be quite good, better than the Mean estimator, especially when the states are close to each other, which is the case in Bayesian tomography.

We also present analytic upper bounds for the maximum value of the average fidelity $f(\sigma)$ for any state $\sigma \in \operatorname{D}(\mathcal H)$:
\begin{equation}    \label{eq:informal_bounds}
    f(\sigma ) \leq \sqrt{\sum_{i,j=1}^{n} p_i p_j \operatorname{F}(\rho_i,\rho_j)} \leq \sqrt{f(\sigma_\text{M})},
\end{equation}
where $\sigma_\text{M} = \sum_{i=1}^{n} p_i \rho_i$ is the mean of the distribution $(\mathsf R,p)$. We call these bounds \textit{Product bound} and \textit{Average bound} respectively.

Note that the average fidelity $f(\sigma')$ of the Commuting estimator $\sigma'$ is a lower bound on the maximum average fidelity. Bringing the Product bound into the picture, we have lower and upper bounds for optimal average fidelity $f(\sigma_\sharp)$:
\begin{equation}
    f(\sigma') \leq f(\sigma_\sharp) \leq \sqrt{\sum_{i=1}^{n} p_i p_j \operatorname{F}(\rho_i,\rho_j) }.
\end{equation}
Both the upper and lower bounds coincide with optimal average fidelity when all the states in the ensemble commute pairwise.

\section{Preliminaries} \label{sec:Prelim}

In this section, we introduce the mathematical preliminaries and notations used in our work. We use uppercase calligraphic letters $\mathcal H, \mathcal X, \mathcal Y,\mathcal Z$ to denote complex Euclidean (Hilbert) spaces and uppercase serif letters $\mathsf A, \mathsf B, \mathsf R$ to denote sets. By $\mathcal X = \mathbb{C}^d$, we mean that $\mathcal X$ is a $d$-dimensional Hilbert space. We use $\operatorname{L}(\mathcal X, \mathcal Y)$  to denote linear operators from Hilbert spaces $\mathcal X$ to $\mathcal Y$ and $\operatorname{L}(\mathcal X)$ to denote linear operators from $\mathcal X$  to $\mathcal X$. Notations for other specific classes of matrices are shown Table~\ref{tab:notation}.

\begin{table}
\centering
\begin{tabular}{|l|l|}
\hline
\textbf{symbol} & \textbf{definition} \\
\hline \hline
$\operatorname{Herm}(\mathcal X )$ & \emph{Hermitian} matrices on $\mathcal{X}$ (i.e. $X^*=X$), \\
$\operatorname{Pos}(\mathcal X )$ & \emph{positive semidefinite} matrices on $\mathcal{X}$ (i.e.\ $X \geq 0$), \\
$\operatorname{Pd}(\mathcal X )$ & \emph{positive definite} matrices on $\mathcal{X}$ (i.e.\ $X >0$), \\
$\operatorname{D}(\mathcal X )$ & \emph{density} matrices on $\mathcal{X}$ (i.e.\ $\rho \geq 0$ and $\operatorname{Tr}(\rho)=1$),\\
$\operatorname{U}(\mathcal X )$ & \emph{unitary} matrices on $\mathcal{X}$ (i.e.\ $UU^*=U^*U=\mathds{1}_{\mathcal{X}}$). \\
\hline
\end{tabular}
\caption{Notational conventions regarding different types of matrices. Recall that $\mathcal{X}=\mathbb{C}^d$ for some dimension $d$ (e.g. $d=2^t$ for $t$ qubits).}
\label{tab:notation}
\end{table}

When the identity of the underlying space is not important, we use $P \geq 0$ to denote the matrix $P$ is positive semidefinite. Similarly we use $P > 0$ to denote $P$ is positive definite. For any positive semidefinite matrix $P \geq 0$, we use the notation $P^{1/2}$ or $\sqrt{P}$ to denote the unique positive semidefinite matrix such that $P^{1/2}P^{1/2} = P$. We also note that a matrix $P$ is positive semidefinite if and only if it can be written as $B^*B$ for some matrix $B$. For two matrices $A,B$, the statement $A \geq B$ is equivalent to $A - B \geq 0$. Similarly, $A > B$ is equivalent to $A - B > 0$.

Two norms we will frequent upon are the \textit{trace norm} and the \textit{spectral norm}, both of which come under a class of matrix norms called Schatten $p$-norms~\cite{Watrous2018}. For any operator $A \in \operatorname{L}(\mathcal X, \mathcal Y)$, the trace norm is defined as
\begin{equation}
    \|A\|_{1} = \operatorname{Tr}\left(\sqrt{A^*A} \right),
\end{equation}
while the spectral norm is defined as
\begin{equation}
    \|A\| = \max \{\|Au\| : u \in \mathcal X , \|u \| \leq 1\}.
\end{equation}
Any operator $A$ with sub-unit spectral norm, $\|A\| \leq 1$, is called a \textit{contraction}. 

\subsection{Polar decomposition}
Throughout this paper, we will use polar decompositions of full rank matrices. Polar decomposition allows us to write an any arbitrary matrix as the product of a unitary and a positive (semi)definite matrix. Let $A \in \operatorname{L}(\mathcal X) $ be an arbitrary full rank square matrix. Then the (right) polar decomposition of $A$ is given by
\begin{equation} \label{eq:rightPD}
    A = U |A|
\end{equation}
where $U \in \operatorname{U}(\mathcal H)$ is a unitary matrix and $|A| = \sqrt{A^*A} \in \operatorname{Pd}(\mathcal H)$ is positive definite. Equivalently, the left polar decomposition can be written as
\begin{equation} \label{eq:leftPD}
    A = |A^*|U.
\end{equation}
Note that the $U$ in \eqref{eq:rightPD} and \eqref{eq:leftPD} are the same. $U$ is unique when $A$ is full rank.

In this article, we will frequent upon cases where $A$ has the special form of being the product of two positive definite matrices: $A = PQ$ where $P,Q \in \operatorname{Pd}(\mathcal H)$. Let $A = \rho^{1/2} \sigma^{1/2}$ for full rank states $\rho, \sigma \in \operatorname{D}(\mathcal H)$. We know that 
\begin{equation}
    \operatorname{F}(\rho,\sigma) = \operatorname{Tr} \left(\sqrt{\sigma^{1/2} \rho \sigma ^{1/2}} \right)=\operatorname{Tr}\left(\sqrt{A^*A} \right) = \|A\|_1.
\end{equation}
The variational property of trace norm states that for any matrix $M$, 
\begin{equation}  \label{eq:VarTraceNorm}
    \max_{U: \|U\| \leq 1}|\operatorname{Tr}(UM)| = \operatorname{Tr}(|M|),
\end{equation}
where $U$ is varied over all contractions, i.e., $\|U\| \leq 1$. The unitary $U$ that saturates the inequality is the inverse of the unitary factor from the polar decomposition of $M$.
Using this property, we have
\begin{equation}
    \max_{\|U\| \leq 1}\operatorname{Tr}\left(\sigma^{1/2} \rho^{1/2}U \right) = \operatorname{F}(\rho,\sigma). 
\end{equation}
The optimal operator $U$ is the inverse of the unitary factor from the polar decomposition of $\rho^{1/2} \sigma^{1/2}$, that is, $U = V^*$ where
\begin{equation} \label{eq:UnitaryFactor}
    \rho ^{1/2} \sigma^{1/2} = V \left|\rho ^{1/2} \sigma^{1/2}\right|.
\end{equation}
At the risk of being obvious, but due to the ubiquity of this relation in the following sections, we note that \begin{equation} \label{eq:PolDecomp}
    \left|\rho ^{1/2} \sigma^{1/2}\right| = \sqrt{\sigma^{1/2} \rho \sigma^{1/2} } = U\rho^{1/2} \sigma^{1/2}  = \sigma^{1/2}\rho^{1/2} U^*.
\end{equation}

\subsection{Semidefinite programs}

Semidefinite programming is a field of convex optimization~\cite{Boyd2004,Barvinok2002}, where we are concerned with optimizing a linear function over the intersection of the cone positive semidefinite matrices with an affine space. Formally, a semidefinite program (SDP) is specified by a triple $(\Phi, A,B)$~\cite{Watrous2018} where $\Phi: \operatorname{L}(\mathcal X)  \to \operatorname{L}(\mathcal Y)$  is a Hermitian preserving linear map between Hilbert spaces $\mathcal X$  and $\mathcal Y$  while $A \in \operatorname{Herm}(\mathcal X)$  and $B \in \operatorname{Herm} (\mathcal Y)$ are Hermitian operators.

Given a triple $(\Phi, A, B)$ of the above form, we may associate two optimization problems with it, which we call the \textit{primal} and \textit{dual problems}. 

\begin{equation}
\begin{aligned}
&\text{Primal problem }\\
\text{maximize : }&\langle A,X\rangle,\\
\text{subject to : } 
&\Phi(\operatorname{X}) = B, \\
&X \in \operatorname{Pos}(\mathcal X).\\
\end{aligned}
\end{equation}

\begin{equation}
\begin{aligned}
& \text{Dual problem}\\
\text{minimise : } &\langle B,Y\rangle,\\
\text{subject to : } 
&\Phi^*(\operatorname{Y}) \geq A, \\
&Y \in \operatorname{Herm}(\mathcal Y).\\
\end{aligned}    
\end{equation}

The sets of all operators that satisfy the respective constraints are called \textit{primal feasible} set $\mathsf A$ and \textit{dual feasible} set $\mathsf B$:
\begin{equation}
\begin{aligned}
\mathsf A &= \{X \in \operatorname{Pos}(\mathcal{X}) : \Phi(\operatorname{X}) = B\}, \\
\mathsf B &= \{Y \in \operatorname{Herm}(\mathcal Y) : \Phi^*(\operatorname{Y}) \geq A \}.
\end{aligned}
\end{equation}
The \textit{primal optimum} $\alpha$ and \textit{dual optimum} $\beta$ are then defined as 
\begin{equation}
\begin{aligned}
    \alpha &= \sup \{ \langle A, X \rangle : X \in \mathsf A\}, \\ 
    \beta &= \inf \{\langle B,Y \rangle : Y \in \mathsf B\}. 
\end{aligned}
\end{equation}
In the case that $\mathsf A = \varnothing$   or $\mathsf B = \varnothing$, we say that $\alpha = -\infty$ or $\beta  = \infty$ respectively.

With semidefinite programs, there exist certain notions of duality, which manifest as \textit{weak duality} and \textit{strong duality}. The property of weak duality, which holds for all semidefinite programs, is that the primal optimum is always bounded above by the dual optimum: $\alpha \leq \beta$.  Strong duality describes the situation where the inequality is saturated~(see Fig. \ref{Fig:Flowchart}, subfigure 3 for an illustration). Note that not all SDPs admit strong duality. Slater's condition is a sufficient condition for strong duality. 

\begin{theorem} [Slater's theorem for semidefinite programs \cite{Watrous2018}] Let $\mathcal X$  and $\mathcal Y$  be complex Euclidean spaces, let $\Phi: \operatorname{L}(\mathcal X)\to \operatorname{L}(\mathcal Y )$  be a Hermitian-preserving map, and let $A \in \emph{Herm}(\mathcal X)$  and $B \in \emph{Herm}(\mathcal Y )$  be Hermitian operators. Letting $\mathsf A,\mathsf B, \alpha,$ and $\beta$ be defined as above for the semidefinite program $(\Phi, A,B)$, the following two implications hold true:
\begin{enumerate}
    \item If $\alpha$ is finite and there exists a Hermitian operator $Y \in \emph{Herm}(\mathcal Y)$  such that $\Phi^*(\operatorname{Y}) > A$, then $\alpha=\beta$, and moreover there exists a primal-feasible operator $X \in \mathsf A$  such that $\langle A,X\rangle = \alpha$.
\item If $\beta$ is finite and there exists a positive definite operator $X \in \emph{Pd}(\mathcal X )$ such that $\Phi(\operatorname{X})=B$, then $\alpha= \beta$, and moreover there exists a dual-feasible operator $Y \in \mathsf B$ such that $\langle B,Y\rangle = \beta$.
\end{enumerate}
\end{theorem}

One must note that though the satisfaction of either of the above two statements implies strong duality $\alpha = \beta$, only if both are satisfied can we say that there exist feasible operators $X$ and $Y$ achieving this optimum value. In the case where there exists such optimal primal and dual feasible operators $X$ and $Y$ with $\langle A,X \rangle  = \alpha = \beta = \langle B, Y\rangle $, there exists a certain relation between these two operators namely \textit{complementary slackness}. 

\begin{theorem}
[Complementary slackness for semidefinite programs \cite{Watrous2018}] Let $\mathcal X$  and $\mathcal Y$  be complex Euclidean spaces, let $\Phi: \operatorname{L}(\mathcal X)\to \operatorname{L}(\mathcal Y )$ be a Hermitian-preserving map, and let $A \in \emph{Herm}(\mathcal X)$  and $B \in \emph{Herm}(\mathcal Y)$  be Hermitian operators. Let $\mathsf A$  and $\mathsf B$  be the primal-feasible and dual-feasible sets associated with the semidefinite program $(\Phi, A, B)$, and suppose that $X \in \mathsf A$  and $Y \in \mathsf B$  are operators satisfying $\langle A,X \rangle = \langle B,Y\rangle$. It holds that
\begin{equation}
\Phi^*(Y)X = AX \quad \text{and} \quad \Phi(X) Y = BY.
\end{equation}
\end{theorem}

\subsection{Semidefinite program for fidelity} \label{sec:WatrousSDP}

The fidelity $\operatorname{F}(P,Q)$  between two states (or more generally between two arbitrary positive semidefinite matrices) $P, Q \in \operatorname{D}(\mathcal H)$  can be formulated as a semidefinite program $(\Phi, A, B)$, see e.g.\ \cite[Theorem~3.17]{Watrous2018}. 
This is a restatement of the variational characterization of trace norm as a semidefinite program. For $\mathcal X = \mathcal Y = \mathcal H \oplus \mathcal H$, define the Hermitian matrices $A$ and $B$ as
\begin{equation}
    A = \frac12 \begin{pmatrix}
0 & \mathds{1}_\mathcal X  \\
\mathds{1}_\mathcal X & 0
\end{pmatrix} \in \operatorname{Herm}(\mathcal X)  \quad \text{and} \quad 
B = \frac12
\begin{pmatrix}
P & 0 \\
0 & Q
\end{pmatrix} \in \operatorname{Herm}(\mathcal Y),
\end{equation}
and the Hermitian preserving map $\Phi$ as
\begin{equation}
    \Phi \begin{pmatrix}
C & \cdot \\
\cdot & D
\end{pmatrix} = 
\frac12 \begin{pmatrix}
C & 0 \\
0 & D
\end{pmatrix}. 
\end{equation}
Here $\cdot$ represents arbitrary block sub-matrices that are zeroed out by the action of $\Phi$. The semidefinite program $(\Phi,A,B)$ as defined above achieves strong duality and has the optimal value $\alpha = \beta = \operatorname{F}(P,Q)$. The primal optimization problem can be written as
\begin{equation} \label{SDP:FidSDP}
\begin{aligned}
&\text{Primal problem}\\
\text{maximize : } & \frac12 \operatorname{Tr}(X) + \frac12 \operatorname{Tr}(X^*) = \Re(\operatorname{Tr}(X))\\
\text{subject to : } 
& \begin{pmatrix}
    P & X \\
    X^* & Q
\end{pmatrix} \geq 0. 
\end{aligned}
\end{equation}
Here we use $\Re$ to denote the real part of the argument. This SDP always satisfies the first of Slater's conditions. When the states $P,Q$  are full rank, both of Slater's conditions hold. If both of Slater's conditions are satisfied, complementary slackness necessarily follows. The following relation between the constraint is leveraged in the construction of the above SDP~\cite[Lemma~3.18]{Watrous2018}, \cite[Proposition~1.3.2]{BhatiaPD}:
\begin{equation} \label{eq:SDPContraction}
    \begin{pmatrix}
    P & X \\
    X^* & Q
\end{pmatrix} \geq 0 \iff X = P^{1/2}UQ^{1/2} : U \in \operatorname{L}(\mathcal H),  \|U\| \leq 1. 
\end{equation}
It follows from the variational property of the trace norm \eqref{eq:UnitaryFactor} that the operator $U = V^*$ where $V$ is the unitary factor from the polar decomposition of the matrix $Q^{1/2} P^{1/2}$:
\begin{equation}
    Q ^{1/2} P ^{1/2} = V \left| Q ^{1/2} P ^{1/2} \right|. 
\end{equation}

\section{Formal statement of results} \label{sec:Results}

\subsection{SDP for optimal average fidelity} \label{sec:OptFidSDP}
The semidefinite program for optimal average fidelity is a generalization of Watrous's SDP for fidelity. We begin with formally defining the SDP whose primal optimum is the optimal average fidelity \eqref{eq:maximprob}. Moreover, this SDP also provides the optimal state.
\begin{definition}[SDP for optimal average fidelity] \label{def:SDP}
Let $\mathsf R = \{\rho_1, \ldots,\rho_n\}  \subset \operatorname{D}(\mathcal H)$ be a collection of quantum states and $p \in \Delta_n$ be a probability vector over $\mathsf R$. Let $\mathcal X  = \bigoplus_{i=1}^{n+1} \mathcal H$ and $\mathcal Y = \bigoplus_{i=1}^n \mathcal H \oplus \mathbb{C}$ be Hilbert spaces with dimensions $(n+1)d$ and $nd + 1$ respectively. Define the semidefinite program $(\Phi, A, B)$ as
\begin{equation} \label{eq:BayesSDPAB}
    \begin{aligned}
    A  = \frac12 
    \begin{pmatrix}
        0 &\cdots & 0 &p_1 \mathds{1}_\mathcal H  \\
        \vdots & \ddots & \vdots & \vdots \\
        0 & \cdots & 0 & p_n \mathds{1}_\mathcal H\\
        p_1\mathds{1}_\mathcal H  & \cdots &p_n \mathds{1}_\mathcal H & 0 
    \end{pmatrix}  \in \operatorname{Herm} (\mathcal X), \\     
    B = 
    \begin{pmatrix}
        \rho_1 & 0 &\cdots & 0 \\
        0 & \ddots & &\vdots  \\
        \vdots && \rho_n & 0\\
        0 &\cdots & 0 & 1
    \end{pmatrix} \in \operatorname{Herm}(\mathcal Y),
    \end{aligned}
    \end{equation}
    and the Hermitian preserving map $\Phi : \operatorname{L}(\mathcal X) \to \operatorname{L}(\mathcal Y)$ which acts as 
    \begin{equation}
    \Phi\begin{pmatrix}
    P_1 & \cdot &\cdots & \cdot \\
    \cdot & \ddots & &\vdots  \\
    \vdots && P_n & \cdot\\
    \cdot &\cdots & \cdot & Q
    \end{pmatrix} =  \begin{pmatrix}
    P_1 & 0 &\cdots & 0 \\
    0 & \ddots & &\vdots  \\
    \vdots && P_n & 0\\
    0 &\cdots & 0 & \operatorname{Tr}(Q) 
    \end{pmatrix}.
    \end{equation}
\end{definition}
The map $\Phi$ acts like the identity on the $d \times d$ block diagonal submatrices, except the last $d\times d$ submatrix which is traced. Every other element is zeroed out. The adjoint map (of $\Phi$) $\Phi^* : \operatorname{L}(\mathcal Y) \to \operatorname{L}(\mathcal X)$ has its action defined as
\begin{equation}
    \Phi\begin{pmatrix}
    P_1 & \cdot &\cdots & \cdot \\
    \cdot & \ddots & &\vdots  \\
    \vdots && P_n & \cdot\\
    \cdot &\cdots & \cdot & q
    \end{pmatrix} =  \begin{pmatrix}
    P_1 & 0 &\cdots & 0 \\
    0 & \ddots & &\vdots  \\
    \vdots && P_n & 0\\
    0 &\cdots & 0 & q \mathds{1}_\mathcal H 
    \end{pmatrix},
\end{equation}
where $q \in \mathbb{C}$.


We now discuss how the optimal average fidelity $\max_{\sigma \in \operatorname{D}(\mathcal H) } f(\sigma)$ is an upper bound for the primal objective function $\langle A, \Chi \rangle $ of the above SDP. 

\begin{lemma}
The primal objective function of the semidefinite program $(\Phi,A,B)$ from Definition \ref{def:SDP} is bounded above by the optimal average fidelity $\max_{\sigma \in \operatorname{D}(\mathcal H) } f(\sigma)$.
\end{lemma}

\begin{proof}
    Under the constraint $\Phi(\Chi) = B$, any primal feasible $\Chi \in \mathsf A$ must have the form
    \begin{equation} \label{eq:primalfeasible}
        \Chi  = \begin{pmatrix}
        \rho_1 & R_{12} &\cdots & R_{1n} & X_1 \\
        R_{12}^* & \rho_2 & \cdots & R_{2n} & X_2  \\
        \vdots & \vdots & \ddots & \vdots & \vdots \\
        R_{1n}^* & R_{2n}^* & \cdots & \rho_n & X_n\\
        X_1^* & X_2^* &\cdots & X_n^* & \sigma
        \end{pmatrix} \geq 0.
    \end{equation}  
    Here, each of the block submatrices $\rho_i, \sigma, X_i, \text{and } R_{ij} \in \operatorname{L}(\mathcal H)$ are square matrices, with $\rho_i$ and $\sigma$ being positive semidefinite as well, for all $i,j \in \{1,\ldots,n\}$. Let us briefly note what these square matrices are. The $\rho_i$s are the fixed quantum states in the ensemble $\mathsf R$ and $\sigma \in \operatorname{D}(\mathcal H)$ is the actual objective matrix which is varied over the set of all $d$-dimensional quantum states $\operatorname{D}(\mathcal H)$. $R_{ij}$s are matrices whose specific form we shall discuss later and the form of the matrices $X_i$s are discussed next. 
    
    The positivity of $\Chi$ necessarily implies positivity of each principle sub-block. In particular,
    \begin{equation}
        M_i = \begin{pmatrix}
        \rho_i & X_i \\
        X_i^* & \sigma
        \end{pmatrix} 
        \geq 0 \quad \text{for all $i=1,\ldots,n$}.
    \end{equation} 
    By \eqref{eq:SDPContraction}, we have $M_i \geq 0$ if and only if $X_i = \rho_i^{1/2} U_i \sigma^{1/2}$ for some $U_i$ with $\|U_i\| \leq 1$. Hence, by the variational property of the trace norm~(see \eqref{eq:VarTraceNorm}), we have that $M_i \geq 0$ necessarily implies $\Re( \operatorname{Tr}(X_i)) \leq \operatorname{F}(\rho_i,\sigma)$. Therefore for any primal feasible point $\Chi$, we have the following inequality involving the objective function $\langle A,\Chi \rangle$:
    \begin{equation} \label{eq:FeasibleInequality}
        \langle A,\Chi \rangle =\sum_{i=1}^{n} p_i \Re(\operatorname{Tr}(X_i)) \leq \sum_{i=1}^{n} p_i \operatorname{F}(\rho_i,\sigma) \leq \max_{\sigma \in \operatorname{D}(\mathcal H) } f(\sigma) = f(\sigma_\sharp),
    \end{equation}
    where $\sigma_{\sharp}$ is the optimal state. That is, the value of the objective function $\langle A, \Chi \rangle$ is bounded above by the optimal average fidelity.
    \end{proof}
    Essentially, the maximization happens at two levels: the first level being that each $X_i$ is varied to maximize the real part of its trace, which is constrained by the value of the variable $\sigma$ (along with the value of each fixed $\rho_i$) and the degree of freedom we allow to $\sigma$ defines the second level of maximization.

    We now prove, in two different ways, that the inequality \eqref{eq:FeasibleInequality} is saturated. The first proof, which makes use of the form of $\Chi$ and the fact that we are optimizing over the closed and bounded set of density matrices, culminates in the following theorem.  
    \begin{theorem}
    For any ensemble $(\mathsf R, p)$, all the inequalities in \eqref{eq:FeasibleInequality} is saturated by the semidefinite program $(\Phi, A, B)$.
    \end{theorem}

    \begin{proof}
        We first establish that for any state $\sigma \in \operatorname{D}(\mathcal H)$, there exists a primal feasible point $\Chi(\sigma) \geq 0$ such that $\langle A, \Chi (\sigma) \rangle = f(\sigma)$. To see this, fix $\sigma \in \operatorname{D}(\mathcal H)$ arbitrary and consider the $(n+1)d \times d$ matrix $Z = Z(\sigma)$ of the form
    \begin{equation}
    Z = \begin{pmatrix}
    \rho_1^{1/2} U_1 \\ \rho_2^{1/2} U_2 \\ \vdots  \\ \rho_n^{1/2} U_n \\ \sigma^{1/2}.
    \end{pmatrix}
    \end{equation}
    Here, each $U_i$ is the \textit{optimal} unitary that maximizes $\Re\left(\operatorname{Tr}\left(\rho_i^{1/2}U_i \sigma^{1/2}\right) \right)$. Consider
    \begin{equation} \label{eq:ChiSigma}
    \begin{aligned}
    \Chi(\sigma) = 
        ZZ^* &= \begin{pmatrix}
            \rho_1 & \rho_1^{1/2}U_1 U_2^*\rho_2^{1/2} &\cdots & \rho_1^{1/2}U_1 U_n^*\rho_n^{1/2} & \rho_1^{1/2}U_1 \sigma^{1/2} \\
            \rho_2^{1/2}U_2 U_1^*\rho_1^{1/2} & \rho_2 & \cdots & \rho_2^{1/2}U_2 U_n^*\rho_n^{1/2} & \rho_2^{1/2}U_2 \sigma^{1/2}  \\
            \vdots & \vdots & \ddots & \vdots & \vdots \\
            \rho_n^{1/2}U_n U_1^*\rho_1^{1/2} & \rho_n^{1/2}U_n U_2^*\rho_2^{1/2} & \cdots & \rho_n & \rho_n^{1/2}U_n \sigma^{1/2}\\
            \sigma^{1/2}U_1^* \rho_1^{1/2} & \sigma^{1/2}U_2^* \rho_2^{1/2} &\cdots & \sigma^{1/2}U_n^* \rho_n^{1/2} & \sigma
            \end{pmatrix}. 
    \end{aligned}
    \end{equation}
    We see that $\Chi(\sigma) \geq 0$ for any state $\sigma$ and it follows that for such a choice of $\Chi(\sigma)$, we have $\langle A,\Chi (\sigma)\rangle = f(\sigma)$. 
    
    Noting that the set $\operatorname{D}(\mathcal H) $ is closed and bounded, it follows that the (continuous) function $f(\sigma)$ achieves its maximum on $\operatorname{D}(\mathcal H)$. Therefore there exists an optimal $\sigma_\sharp \in \operatorname{D}(\mathcal H)$ achieving the optimal average fidelity. The optimal primal feasible point $\Chi(\sigma_\sharp) $ such that $ \langle A, \Chi(\sigma_\sharp) \rangle  = f(\sigma_\sharp)$ can be then constructed by \eqref{eq:ChiSigma}.
    \end{proof}
    
The second proof makes use of Slater's condition. As we show in Theorem \ref{Th:Slater}, SDP $(\Phi, A, B)$ always satisfies the first of Slater's conditions, thereby ensuring strong duality and the existence of a primal feasible point which achieves optimal value. Moreover, when both of Slater's conditions are satisfied, which happens when all the states in the ensemble are full rank, complementary slackness relations hold.

We formulate the satisfaction of Slater's conditions by the SDP $(\Phi, A, B)$ as a separate theorem which also accounts for complementary slackness. The following theorem shows that the above SDP exhibits strong duality and, when all the states in $\mathsf R$ are full rank, complementary slackness. The fixed point equation, and thereby the fixed-point algorithm for the optimal state, will then arise from complementary slackness.

\begin{theorem} \label{Th:Slater}
The semidefinite program $(\Phi, A, B)$ exhibits strong duality. Moreover,
complementary slackness holds if and only if all the states in $\mathsf{R}$ are full rank.
\end{theorem}
\begin{proof}    
    Recall that the first Slater's condition requires the primal constraint to be satisfied $(\mathsf A \neq \varnothing)$ and the dual constraint to be strictly satisfied $(\Phi^*(\operatorname{Y}) > A)$. Note that for any quantum state $\rho \in \operatorname{D}(\mathcal H)$, the matrix
\begin{equation}
    \Chi = \begin{pmatrix}
    \rho_1 & 0 &\cdots & 0 \\
    0 & \ddots & &\vdots  \\
    \vdots && \rho_n & 0\\
    0 &\cdots & 0 & \rho
    \end{pmatrix},
\end{equation}
is a primal feasible point as it satisfies $\Phi(\operatorname{X}) = B$, thereby ensuring $\mathsf A \neq \varnothing$. For the second part, consider the identity matrix 
    \begin{equation}
        \operatorname{Y} = \begin{pmatrix}
    \mathds{1}_\mathcal H  & 0 &\cdots & 0 \\
    0 & \ddots & &\vdots  \\
    \vdots && \mathds{1}_\mathcal H & 0\\
    0 &\cdots & 0 & 1
    \end{pmatrix} \in \operatorname{Herm}(\mathcal Y),
    \end{equation}
    and note that
    \begin{equation}
        \Phi^*(\operatorname{Y}) - A = \begin{pmatrix}
    \mathds{1}_\mathcal H  & 0 &\cdots & -\frac12p_1 \mathds{1}_\mathcal H \\
    0 & \ddots & &\vdots  \\
    \vdots && \mathds{1}_\mathcal H & -\frac12p_n \mathds{1}_\mathcal H\\
    -\frac12p_1 \mathds{1}_\mathcal H &\cdots & -\frac12p_n \mathds{1}_\mathcal H & \mathds{1}_\mathcal H
    \end{pmatrix}  > 0.
    \end{equation}
     This is by  \cite[Theorem 6.1.10]{Hornmatrix}, which states that if a Hermitian matrix $M$ with strictly positive diagonal entries is strictly diagonally dominant i.e., $|M(i,i)| > \sum_j|M(i,j)|$ for all rows $i$, then $M$ is positive definite. Since $p_i < 1$ for any $i \in \{1,\ldots,n\},$  we see that $\Phi^*(\operatorname{Y})-A$ satisfies this criterion and therefore is positive definite. Of course, there could be the case where $p_i = 1$ for some $i$ and $0$ for all $j\neq i$. However in such a case the optimization problem is trivially solved by choosing the optimal state to be $\sigma_\sharp = \rho_i$ and we get $f(\sigma_\sharp) = 1$. We do not consider such a deterministic scenario.
    
    Hence, the semidefinite program $(\Phi, A, B)$ satisfies the first Slater condition which in turn ensures strong duality. Moreover, this optimum $\alpha$ is always achieved for some $\Chi_\sharp \in \mathsf A$. The second condition, in contrast, need not always hold. One sees that if any state $\rho_i$ is \emph{not} full rank, then there exists no $\Chi \in \operatorname{Pd}(\mathcal X )$ satisfying $\Phi(\operatorname{X}) = B$, as $B$ itself is rank deficient and therefore complementary slackness will not hold. However, for the case where all states $\{\rho_i\}$ are full rank, then the second condition is also satisfied which in turn leads to complementary slackness. To show that $\mathsf B$ is non-empty, we may use the same $\operatorname{Y}$ as above and $\Chi$ as defined above satisfies $\Phi(\Chi) = B$ while being positive definite for any full rank quantum state $\rho$. 
    
    Therefore, if, and only if, all the states $\{\rho_i\}$ are full rank, there exists primal and dual optimal operators $\Chi_\sharp \in \mathsf A$  and $\operatorname{Y}_\sharp \in \mathsf B$ such that
    \begin{equation}
    \langle A,\Chi_\sharp\rangle = \alpha = \beta =\langle B,\operatorname{Y}_\sharp\rangle.
    \end{equation}
    Such a condition implies complementary slackness between primal and dual optimal points which manifests as  
    \begin{equation} \label{eq:compslack}
    \Phi^*(\operatorname{Y}_\sharp) \Chi _\sharp = A\Chi_\sharp. 
    \end{equation}
\end{proof}    

The second proof of the saturation of inequality \eqref{eq:FeasibleInequality} follows from the fact that the semidefinite program $(\Phi,A,B)$ exhibits strong duality. We now focus on \eqref{eq:compslack}, which allows us to infer insights about the properties of the optimal state $\sigma_\sharp$, including a fixed point equation. This is formally discussed in the next theorem.

\begin{theorem} [Fixed point equation for optimal state]
If all the states $\{\rho_1,\ldots, \rho_n\}$ are full rank, the following equation holds:
\begin{equation} \label{eq:FPOri}
\sigma_\sharp = \frac1{f(\sigma_\sharp)} \sum_{i=1}^{n} p_i \rho_i^{1/2} U_i \sigma_\sharp^{1/2} = 
\frac1{f(\sigma_\sharp)} \sum_{i=1}^{n} p_i  \sigma_\sharp^{1/2} \rho_i^{1/2} U_i
, 
\end{equation}
for certain unitaries $U_i \in \emph{U}(\mathcal H)$  and optimal state $\sigma_\sharp$. This, in turn, implies
\begin{equation} \label{eq:FPeqn} 
    \sigma_\sharp = \Gamma \left( \sum_{i=1}^{n} p_i \sqrt{\sigma_\sharp^{1/2} \rho_i \sigma_\sharp^{1/2}}
\right),
\end{equation}
where the map $\Gamma$ is defined as $\Gamma(A) = A / \emph{Tr}(A)$. 
\end{theorem}
\begin{proof}
    Let $\Chi_\sharp \in \operatorname{Pos}(\mathcal X)$   and $\operatorname{Y}_\sharp \in \operatorname{Herm} (\mathcal Y )$ be the optimal primal and dual points satisfying
    \begin{equation}
        \langle A,\Chi_\sharp \rangle = \alpha = \beta = \langle B, \operatorname{Y}_\sharp \rangle  = f(\sigma_\sharp).
    \end{equation}
    Decompose $\Chi_\sharp$ and $\Phi^*(\operatorname{Y}_\sharp)$ as
    \begin{equation}
         \Chi_\sharp  = \begin{pmatrix}
         \rho_1 & R_{12} &\cdots & R_{1n} & X_1 \\
         R_{21} & \rho_2 & \cdots & R_{2n} & X_2  \\
         \vdots & \vdots & \ddots & \vdots & \vdots \\
         R_{n1} & R_{n2} & \cdots & \rho_n & X_n\\
         X_1^* & X_2^* &\cdots & X_n^* & \sigma_\sharp
         \end{pmatrix}, 
    \Phi^*(\operatorname{Y}_\sharp) = \begin{pmatrix}
    Y_1 & 0 &\cdots &0 & 0 \\
    0 & Y_2 &\cdots &0 & 0 \\
    \vdots & \vdots &\ddots &\vdots & \vdots \\
    0 & 0 &\cdots &Y_n & 0 \\ 
    0 & 0 &\cdots &0 & 
    z \mathds1_\mathcal H 
    \end{pmatrix},
    \end{equation}
    where $R_{ij} = R_{ji}^*$. By complementary slackness, we have $\Phi^*(\operatorname{Y}_\sharp) \Chi_\sharp = A \Chi_\sharp$, which is equivalent to demanding
    \begin{equation} \label{eq:CompSlack}
    \begin{aligned}
        &\begin{pmatrix} 
        Y_1\rho_1 & Y_1R_{12} &\cdots & Y_1R_{1n} & Y_1X_1 \\
        Y_2R_{21} & Y_2\rho_2 & \cdots & Y_2R_{2n} & Y_2X_2  \\
        \vdots & \vdots & \ddots & \vdots & \vdots \\
        Y_nR_{n1} & Y_nR_{n2} & \cdots & Y_n\rho_n & Y_n X_n \\
        zX_1^* & zX_2^* &\cdots & zX_n^* & z\sigma_\sharp
        \end{pmatrix} = 
        \frac12
        &\begin{pmatrix}
        p_1X_1^* & p_1X_2^* &\cdots & p_1X_n^* & p_1 \sigma_\sharp \\
        p_2X_1^* & p_2X_2^* &\cdots & p_2X_n^* & p_2 \sigma_\sharp \\
        \vdots & \vdots &\ddots &\vdots & \vdots \\
        p_nX_1^* & p_nX_2^* &\cdots & p_nX_n^* & p_n \sigma_\sharp \\ 
        S_1 & S_2  & \cdots & S_n & \sum_{i=1}^{n} p_i X_i
        \end{pmatrix},
    \end{aligned}
    \end{equation}
    where $S_i = p_i \rho_i + \sum_{j=1, j \neq i}^{n} p_jR_{ji}$. 
    
    Note that $Y_i$ and $\rho_i$ must be full rank for all $i \in \{1,\ldots,n\}$  for complementary slackness to hold. Together with the fact that $\frac12p_i X_i =(\frac12 p_iX_i^*)^*=(Y_i \rho_i)^*= \rho_i Y_i$, this implies that the $X_i$s are also full rank for all $i \in \{1,\ldots,n\}$. This is a consequence of the fact that $\frac12p_iX_i =  \rho_i Y_i$ and that the product of two full rank matrices is also full rank (see \cite[Sec. 0.4.6]{Hornmatrix}). By the additional equality $Y_iX_i =\frac12 p_i \sigma_\sharp$, which comes from the last column of \eqref{eq:CompSlack}, we have $\sigma_\sharp$ to be full rank. It follows from these relations that $R_{ij} = \rho_i^{1/2} U_i U_j^* \rho_j^{1/2}$, which is in agreement with the forms for $R_{ij}$ (as defined in \eqref{eq:primalfeasible}) we got in \eqref{eq:ChiSigma}. The last block matrix equality from \eqref{eq:CompSlack} is of particular importance:
    \begin{equation}
        \sigma_\sharp = \frac{1}{2z} \sum_{i=1}^{n} p_i X_i.
    \end{equation}
    We will later on derive the fixed point algorithm from this relation. Note that for $\langle A,\Chi \rangle $ to achieve the optimum value $f(\sigma_\sharp)$, for every index $i$ we must have $X_i =  \rho_i ^{1/2} U_i \sigma^{1/2}_\sharp$,  where $ U_i = \left|\rho_i ^{1/2} \sigma^{1/2}_\sharp \right| \sigma^{-1/2}_\sharp \rho_i ^{-1/2}$ is the inverse of the unitary factor of the polar decomposition of $\rho_i ^{1/2} \sigma_\sharp^{1/2}$. It follows that $\operatorname{Tr}(X_i) = \operatorname{F}(\rho_i,\sigma_\sharp)$, and tracing both sides, we obtain $2z = f(\sigma_\sharp)$, Hence we have
      \begin{equation} \label{eq:SigmaCS}
        \sigma_\sharp = \frac1{f(\sigma_\sharp)} \sum_{i=1}^{n} p_i \rho_i^{1/2}U_i \sigma^{1/2}_\sharp.
    \end{equation}
    Left and right multiplying by $\sigma_\sharp^{1/2}$ and $\sigma_\sharp^{-1/2}$ respectively on both sides, we get
    \begin{equation} \label{eq:FP1unnorm}
        \sigma_\sharp = \frac1{f(\sigma_\sharp)} \sum_{i=1}^{n} p_i \sigma_\sharp^{1/2} \rho_i ^{1/2} U_i = \frac1{f(\sigma_\sharp)} \sum_{i=1}^{n} p_i \sqrt{\sigma_\sharp^{1/2} \rho_i \sigma_\sharp^{1/2}},
    \end{equation}
    where the last equality comes from \eqref{eq:PolDecomp}. Equivalently, one may write 
    \begin{equation} \label{eq:FPAlgo}
        \sigma_\sharp= \Gamma\left( \sum_{i=1}^n p_i  \sqrt{\sigma_\sharp^{1/2} \rho_i \sigma_\sharp^{1/2}}    \right) = 
    \Gamma\left( \sum_{i=1}^n p_i  \left| \rho_i^{1/2} \sigma_\sharp^{1/2}\right|    \right),
    \end{equation}
    where $\Gamma(A) = A/\operatorname{Tr}(A)$.
    \end{proof}
    We can now use expression \eqref{eq:FPAlgo} to construct an iterative fixed-point algorithm to obtain the optimal state. This expression has the property that after each iteration, we have a convex combination of positive definite operators which is then normalized to a density matrix. Therefore, we never leave the set of density matrices during the iteration process.

\subsection{Iterative algorithm for optimal fidelity estimator}

Expression \eqref{eq:FPAlgo} can be used to construct a fixed point iterative algorithm of the form
\begin{equation} \label{eq:Lambda}
    \sigma_{(k)}  = \Lambda\left(\sigma_{(k-1)}\right) :=  \Gamma \left( 
\sum_{i=1}^{n} p_i \left|\rho_i^{1/2} \sigma^{1/2}_{(k-1)}\right|
\right),
\end{equation}
where $\Gamma(A) = A / \operatorname{Tr}(A)$. We define the repeated action of $\Lambda$ recursively as $\Lambda^k(\rho) = \Lambda(\Lambda^{k-1}(\rho))$. This algorithm is numerically seen to converge to the optimal state $\sigma$ for all choices of full-rank initial states $\sigma_{(0)}$ and is much more tractable than solving the semidefinite program $(\Phi,A,B)$ numerically as we avoid optimizing over complex matrices of dimension $(n+1)d$. Though the fixed point algorithm is numerically seen to always converge to the optimal state for random initializations, a well-motivated ansatz is the Commuting estimator $\sigma'$~\eqref{eq:SigmaComm} which is studied in the next subsection.

A second fixed point algorithm can be constructed from results in Ref.~\cite{Bhatia2019}, where Bhatia \emph{et al.} are concerned with solving the related problem of minimizing average squared Bures distance $d^2_{\operatorname{B}} (P,Q) = \operatorname{Tr}(P+Q) - 2 \operatorname{F}(P,Q)$. That is, given a collection of positive definite matrices $\{A_1,\ldots,A_n\} \subset \operatorname{Pd}(\mathcal H)$ and a probability vector $p \in \Delta_n$ over it, they find the solution to the optimization problem

\begin{equation} \label{eq:BhatiaMinProb}
    \argmin_{B \in \operatorname{Pd}(\mathcal H)} \sum_{i=1}^{n} p_i d^2_{\operatorname{B}}(A_i,B) = \sum_{i=1}^{n} p_i \big(\operatorname{Tr}(A_i+B) - \operatorname{F}(A_i,B) \big).
\end{equation}
The solution $B_\sharp$ to this problem satisfies a similar fixed point equation to ours:
\begin{equation}    \label{eq:BhatiaFPeqn}
    B_\sharp = \sum_{i=1}^{n} p_i \sqrt{B_\sharp^{1/2} A_iB_\sharp ^{1/2}}.
\end{equation}
The fixed point algorithm, which provably converges to $B_\sharp$, is of the form
\begin{equation} \label{eq:FPalgoKappa}
    B_{(k+1)} = \operatorname{K} \left(B_{(k)}\right) = B^{-1/2}_{(k)}\left(\sum_{i=1}^{n} p_i \sqrt{B_{(k)}^{1/2} A_iB_{(k)} ^{1/2}} \right)^2 B_{(k)}^{-1/2},
\end{equation}
which reduces to \eqref{eq:BhatiaFPeqn} at the fixed point $B_\sharp$. We can modify this algorithm to obtain a second fixed point algorithm which we describe formally in the following theorem.

\begin{theorem} [Convergence of $\Omega$ fixed point algorithm]
Let $R = \{\rho_1, \ldots, \rho_n\} \in \operatorname{Pd}(\mathcal H)$ be a collection of full-rank states and $p \in \Delta_n$ be a probability vector. Consider the map $\Omega :\operatorname{Pd}(\mathcal H) \to \operatorname{D}(\mathcal H)$ of the form
\begin{equation}
    \label{eq:Omega}
    \Omega(\sigma) = \Gamma(\operatorname{K}(\sigma)) = \Gamma  \left( \sigma^{-1/2}\left(  \sum_{i=1}^{n} p_i\sqrt{\sigma^{1/2} \rho_i \sigma^{1/2}}\right)^2 \sigma^{-1/2} \right),
\end{equation}
where $\Gamma$ is defined as $\Gamma(A) = A/\operatorname{Tr}(A)$. Then the collection of states $\{\Omega^k(\sigma)\}_{k=1}$, where we define $\Omega^{k}(\sigma) = \Omega(\Omega^{k-1}(\sigma))$, converges to the optimal state $\sigma_\sharp$ for any full-rank initial state $\sigma \in \operatorname{D}(\mathcal H)$.
\end{theorem}
\begin{proof}
    The proof is a simple extension of \cite[Theorem 11]{Bhatia2019}, which proves that the fixed point algorithm $\operatorname{K}$~\eqref{eq:FPalgoKappa} converges to \eqref{eq:BhatiaFPeqn}. By noting that $\Gamma$ is a continuous function (as it is division by a non-zero function), it follows that the fixed point algorithm $\Omega$ converges to $\sigma_\sharp$. 
\end{proof}

We numerically compare the performance of the two fixed point algorithms in terms of the total time taken for convergence in Appendix \ref{app:AddNumerics}. The results indicate they perform similarly in terms of runtime performance. Since the convergence of $\Omega$ fixed-point algorithm can be proven theoretically, we prefer it over $\Omega$ fixed-point algorithm. In the following sections, when we refer to simply `fixed-point algorithm', we mean the $\Omega$ fixed-point algorithm.

\subsection{Heuristic approximations of optimal fidelity estimators}~\label{Sec:CommEstimator}

If all the states in $\mathsf R$ commute pairwise i.e., $[\rho_i,\rho_j] = 0$ for all $i,j \in \{1,\ldots,n\}$, then there exists a simple analytic expression for the optimal state $\sigma_\sharp$. We begin with noting that when all the states commute pairwise, the problem reduces to a \textit{classical} problem, as the problem can be considered in the common eigenbasis where we are now dealing with a collection of probability vectors. Hence we may take that the optimal solution is also a probability vector, or that the optimal state commutes with all $\rho_i \in \mathsf R$. We then have
\begin{equation}
f(\sigma_\sharp)\cdot\sigma_\sharp = \sum_{i=1}^{n} p_i\sqrt{\sigma_\sharp^{1/2}\rho_i \sigma_\sharp^{1/2}} = \sum_{i=1}^{n} p_i\rho_i^{1/2} \sigma_\sharp^{1/2} 
\end{equation}
Multiply both sides by $\sigma_\sharp^{-1/2}$ and squaring, we obtain
\begin{equation} \label{eq:SigmaComm}
\sigma_\sharp = \sigma':= \Gamma \left(\left( \sum_{i=1}^{n}  p_i \rho_i^{1/2}\right)^2 \right),
\end{equation}
which we call the \textit{Commuting estimator} $\sigma'$. This expression agrees with results from \cite{Ferrie2016} that deal with the problem of optimal states for fidelity restricted to commuting states. The expression \eqref{eq:SigmaComm} also serves as an heuristic approximation in the general case. To see this, recall that $f(\sigma_\sharp)\cdot\sigma_\sharp =  \sum_{i=1}^{n} p_i \rho_i^{1/2} U_i \sigma_\sharp^{1/2}$. Since $U_j = \exp(iH_j)$ for some Hermitian $H_j \in \operatorname{Herm} (\mathcal H)$, we may write
\begin{equation}
f(\sigma_\sharp) \cdot \sigma_\sharp=  \sum_{j=1}^{n} p_j \rho_j^{1/2} \left(\mathds{1}_\mathcal H  + iH_j +  \frac12(iH_j)^2 + \cdots  \right)  \sigma_\sharp^{1/2}.
\end{equation}
Taking a $0$th order approximation by ignoring all terms expect $\mathds{1}_\mathcal H$, we obtain $\sigma'$. This approximation would make sense when all the states $\rho_i \in  \mathsf R$ are `close by', which is the case for Bayesian state estimation which is the proposed application of these results. 

\subsection{Upper and lower bounds on optimal average fidelity}
We now present two different upper bounds on the maximum average fidelity achievable by any state for an arbitrary ensemble $(\mathsf R, p)$ (that may include rank-deficient states). We call these bounds the \textit{Average bound} and the \textit{Product bound}, respectively. The Average bound states that the square root of average fidelity of the Mean estimator $\sigma_\text{M} = \sum_{i=1}^{n} p_i \rho_i$ bounds the average fidelity obtained by state $\sigma$ from above:
\begin{equation}
    f(\sigma) = \sum_{i=1}^{n} p_i \operatorname{F}(\rho_i,\sigma) \leq \sqrt{f(\sigma_\text{M})}.
\end{equation}
The Product bound states that for any state $\sigma$,
\begin{equation}
    f(\sigma)  \leq \sqrt{\sum_{i,j=1}^{n} p_i p_j \operatorname{F}(\rho_i, \rho_j)}.
\end{equation}
The Product bound is tighter than the Average bound. These statements are formalized below. We first prove the following lemma which deals with the scenario where all the states in the ensemble are full-rank.
\begin{lemma} \label{Lem:ProdBoundFullRank}
Let $\mathsf R = \{ \rho_1, \ldots, \rho_n\} \subset \operatorname{D}(\mathcal H) $ be a collection of full-rank states and $p \in \Delta_n$ be a probability vector. Then for any state $\sigma \in \operatorname{D}(\mathcal H)$, it holds that
\begin{equation} \label{eq:prodboundsfullrank}
    f(\sigma) \leq  \sqrt{\sum_{i=1}^{n} p_ip_j\operatorname{F}(\rho_i, \rho_j)}.
\end{equation}
When all the states in the ensemble commute pairwise, the inequality is saturated. 
\end{lemma}
\begin{proof}
    By \eqref{eq:FPOri}, we have
\begin{equation}
    f(\sigma_\sharp) \cdot \sigma_\sharp = \sum_{i=1}^{n} p_i \sqrt{\sigma_\sharp^{1/2} \rho_i \sigma_\sharp^{1/2}} = \sum_{i=1}^{n}  p_i  \sigma_\sharp^{1/2} \rho ^{1/2}_i U_i.     
\end{equation}
Squaring and then left and right multiplying by $\sigma_\sharp^{-1/2}$, we have
\begin{equation} \label{Eq:ProdBoundStep1}
\begin{aligned}
f(\sigma_\sharp)^2 \cdot \sigma_\sharp  &=  \sigma_\sharp^{-1/2} \left(\sum_{i,j=1}^{n} p_i p_j \sigma_\sharp^{1/2} \rho ^{1/2}_i U_i U_j^* \rho_j ^{1/2} \sigma_\sharp^{1/2}  \right)\sigma_\sharp^{-1/2}  \\ 
&= \sum_{i,j=1}^{n} p_i p_j \rho_i ^{1/2} U_i U_j^* \rho_j ^{1/2}. 
\end{aligned}
\end{equation}
Tracing both sides, we have
\begin{equation} \label{Eq:ProdBoundStep2}
f(\sigma_\sharp)^2 = \text{Tr}\left( \sum_{i,j=1}^{n}  p_i p_j  \rho_i ^{1/2}U_i U_j^* \rho_j ^{1/2} \right) = \sum_{i,j=1}^{n}  p_i p_j  \Re\left(\operatorname{Tr}\left( \rho_i ^{1/2}U_i U_j^* \rho_j ^{1/2} \right)\right).
\end{equation}
Noting that $\Re\left(\operatorname{Tr} \left(\rho_i ^{1/2} V   \rho_j ^{1/2} \right) \right) \leq  \operatorname{F}(\rho_i,\rho_j)$  for any unitary $V \in \operatorname{U}(\mathcal H)$, it follows that
\begin{equation} \label{Eq:ProdBoundStep3}
    f(\sigma_\sharp)^2 \leq \sum_{i=1}^{n} p_i p_j \operatorname{F}(\rho_i,\rho_j).  
\end{equation}
Taking square root over both sides, we obtain \eqref{eq:prodboundsfullrank}.

To see that the inequality is saturated in the commuting case, note that if $\rho_i$ and $\sigma_\sharp$ commute, we have $U_i = \mathds{1}_\mathcal H$. Hence \eqref{Eq:ProdBoundStep1} reduces to
\begin{equation}
\begin{aligned}
    f(\sigma_\sharp)^2 & = \text{Tr}\left( \sum_{i,j=1}^{n}  p_i p_j  \rho_i ^{1/2} \rho_j ^{1/2} \right) \\ 
    &=\sum_{i,j=1}^{n}  p_i p_j  \text{Tr}\left( \rho_i ^{1/2} \rho_j ^{1/2} \right) =\sum_{i, j=1}^{n} p_i p_j \operatorname{F}(\rho_i,\rho_j), 
\end{aligned}
\end{equation}
where, for commuting states $\rho_i$ and $\rho_j$, we have $\operatorname{F}(\rho_i,\rho_j) = \text{Tr}\left( \rho_i^{1/2} \rho_j^{1/2}\right)$. Taking square root across, inequality saturation of~\eqref{eq:prodboundsfullrank} follows.
\end{proof}

Lemma \ref{Lem:ProdBoundFullRank} bounds optimal average fidelity in the case where all the states in the ensemble are full rank. This can be extended to ensembles with arbitrary (such as rank-deficient) states by continuity, which brings us to the Product bound for arbitrary ensemble. 
\begin{theorem} [Product bound] \label{Thm:ProdBounds}
Let $\mathsf R = \{\rho_1, \ldots, \rho_n\} \in \operatorname{D}(\mathcal H) $ be an arbitrary collection of quantum states and let $p \in \Delta_n$ be a probability vector. Then for any state $\sigma \in \operatorname{D}(\mathcal H)$, it holds that
\begin{equation} \label{eq:ProdBounds}
    f(\sigma) \leq  \sqrt{\sum_{i=1}^{n} p_i p_j\operatorname{F}(\rho_i, \rho_j) }.
\end{equation}
\end{theorem}
\begin{proof}
Let $\mathsf {R}_\epsilon $ denote the ensemble obtained by depolarizing all the states in $\mathsf R$ by a factor of $\epsilon \in [0,1]$: 
\begin{equation}
\mathsf R_\epsilon = \{\rho_i'(\epsilon) =( 1-\varepsilon)\rho_i + \varepsilon \mathds{1}_{\mathcal H}/d: \rho_i \in \mathsf R\}.
\end{equation}
Let $J(\epsilon )$ denote the gap of the Product bound when the ensemble is depolarised by a factor of $\epsilon$:  
\begin{equation} \label{eq:Jeps}
J(\epsilon ) = \sqrt{ \sum_{i,j=1}^{n}  p_i p_j \operatorname{F} \left(\rho_i',\rho_j' \right) } - \sum_{i=1}^{n} \operatorname{F}(\rho_i',\sigma_\sharp),
\end{equation}
where $\sigma_\sharp$ is the optimal state over the ensemble $\mathsf R_\epsilon$ and we've dropped the $\epsilon$ while writing $\rho_i'$ for brevity. 

For any $\epsilon \in (0,1]$, all of the states in $\mathsf R_\epsilon$ are full rank and thereby $J(\epsilon ) \geq0$ (as the Product bound holds for full-rank states by Lemma \ref{Lem:ProdBoundFullRank}). We are interested in the case where $\epsilon = 0$, which implies $\mathsf R_\epsilon = \mathsf R$, thereby the ensemble can now contain rank-deficient states.

Note that both the terms in RHS of \eqref{eq:Jeps} are continuous functions of $\epsilon$  (for $\epsilon \in [0,1]$) as they're both compositions of continuous functions ($\rho_i'$ is continuous in $\epsilon$  and fidelity is continuous in its arguments). Hence we have $J(\epsilon )$ to be continuous in $\epsilon$. 

By noting that $J(\epsilon )$ is continuous in $\epsilon$ along with the fact that $J(\epsilon) \geq 0$ for  $\epsilon \in (0, 1]$, we conclude that $J(\epsilon = 0) \geq 0$. Equivalently, Product bound holds for arbitrary ensembles (that may include rank-deficient states).
\end{proof}

\begin{remark}
For an ensemble $(\mathsf R, p)$, the average fidelity of the Commuting estimator $\sigma'$~\eqref{eq:SigmaComm} and the Product bound~\eqref{eq:ProdBounds} are lower and upper bounds to optimal average fidelity respectively:
\begin{equation}
    f(\sigma') \leq f(\sigma_\sharp) \leq \sqrt{\sum_{i=1}^{n} p_i p_j \operatorname{F}(\rho_i,\rho_j) }.
\end{equation}
These bounds coincide when all the states in the ensemble $\mathsf R$ commute pairwise.
\end{remark}

We now derive the \textit{average bound}, which is a consequence of joint concavity of fidelity.
\begin{corollary}
Let $\mathsf R = \{\rho_1, \ldots, \rho_n\} \in \operatorname{D}(\mathcal H)$ be an arbitrary collection of quantum states, $p \in \Delta_n$ be a probability vector, and $\sigma_{\operatorname{M}} = \sum_{i=1}^{n} p_i \rho_i$ be the Mean estimator. Then for any state $\sigma \in \operatorname{D}(\mathcal H)$, it holds that
\begin{equation} \label{eq:averagebound}
    f(\sigma) \leq  \sqrt{\sum_{i=1}^{n} p_i p_j\operatorname{F}(\rho_i, \rho_j) } \leq \sqrt{f(\sigma_{\operatorname{M}})}.
\end{equation}
\end{corollary}
\begin{proof}
We have
\begin{equation}    \label{eq:averageboundderivation}
\begin{aligned}
f(\sigma_\text{M}) &= \sum_{i=1}^{n} p_i \operatorname{F}(\rho_i,\sigma_\text{M}) \\
&= \sum_{i=1}^{n} p_i \operatorname{F} \left(\rho_i, \sum_{j=1}^{n} p_j \rho_j \right) \\ 
&\geq \sum_{i=1}^{n} p_i p_j  \operatorname{F}(\rho_i, \rho_j),
\end{aligned}
\end{equation}
where the inequality comes from the (joint) concavity of fidelity~\cite[Property 9.2.2]{Wilde2013}. Combining \eqref{eq:ProdBounds} and \eqref{eq:averageboundderivation} and taking square root across, and noting that $f(\sigma) \leq f(\sigma_\sharp)$ for any state $\sigma \in \operatorname{D}(\mathcal H)$, we obtain \eqref{eq:averagebound}.
\end{proof}

\section{Application: tomography} \label{sec:applications}
 In quantum tomography, we aim to reconstruct the state of a quantum system from measurement outcomes. There exists a multitude of different tomography methods such as linear inversion~\cite{DAriano2007}, projected least square~\cite{Guta2020}, maximum likelihood estimation~\cite{Hradil97, Rehacek2001}, hedged maximum likelihood estimation~\cite{Blume2010HMLE}, compressed sensing tomography~\cite{Gross2010, Flammia2012, Kliesch2016}, and Bayesian quantum tomography~\cite{Blume2010BME, Huszar2012, Granade2016}. We now discuss how our results find application in Bayesian quantum tomography.

In tomography or state estimation, we are presented with \textit{measurement data} $\mathsf D = (E_k, m_k)_k$. Each $E_k$ is part of a POVM with $\sum_{k} E_k = \mathds{1}_{\mathcal H}$ and appears in the measurement data $m_k$ times with $m = \sum_k m_k$ denoting the total number of measurements done. The likelihood of obtaining the data $\mathsf D$ given any true state $\rho$ is then computed as 
\begin{equation}
 \mathcal L(\mathsf D|\rho) = \operatorname{Pr}(\mathsf D|\rho) = \frac{m!}{\prod_k (m_k!)} \prod_k \operatorname{Tr}(E_k \rho)^{m_k}.
\end{equation}
In maximum likelihood estimation, we are concerned with finding the state that maximizes this likelihood function:
\begin{equation}
    \sigma_{\text{MLE}} = \argmax_{\sigma \in \operatorname{D}(\mathcal H)} \mathcal L(\mathsf D|\sigma).
\end{equation}
In practical  Bayesian state estimation~\cite{granade_qinfer_2017}, we instead begin with a collection of states $\mathsf R = \{\rho_i\}_{i=1}^n$ called \textit{particles} and a prior distribution $u \in \Delta_n$ over them. Usually the prior is taken to be as uniform as possible, but our results work for arbitrary distributions so this is not a concern here. 
Using Bayes' rule we then compute the posterior distribution $p \in \Delta_n$ as 
\begin{equation}
    p_i \propto u_i \mathcal L(\mathsf D|\rho_i),
\end{equation}
with the probabilities $p_i$ being normalised afterwards. Once we have the posterior distribution, the Bayes estimator is the state which maximizes the posterior average fidelity over the ensemble $(\mathsf R, p)$. 

The Bayes estimator is then reported as an estimate for the true state which generated the measurement data.
The main results of this work tell us, given the posterior distribution, how to compute the Bayes estimator for fidelity, heuristic approximations for it, and bounds for the maximum average fidelity.



\section{Numerical experiments} \label{Sec:Numerics}
We consider four different numerical experiments in this work. First, we look at the performance (runtime) of the different methods to obtain the optimal state over random distributions of full-rank states (Fig.~\ref{Fig:SDPvFP}). We then simulate Bayesian tomography and showcase how our results can improve on state-of-the-art methods (Fig.~\ref{Fig:LambdaTomo}). Finally, we numerically demonstrate the relative tightness of the bounds we derived (Fig.~\ref{Fig:Bounds}). In Appendix~\ref{app:AddNumerics}, we compare the performance of two fixed-point methods $\Lambda$ and $\Omega$~(Fig.~\ref{Fig:FPCompare}).

We consider the fixed-point methods to have converged (i.e., the stopping condition) if the spectral norm of the difference between two consecutive iterations is less than some tolerance $\epsilon$. In Fig.~\ref{Fig:SDPvFP}, we choose $\epsilon = 10^{-4}$. 

We use CVXPY~\cite{CVXPY1, CVXPY2} to solve SDPs numerically. The numerics were done on Google Colab (single-core CPU at 2.3 GHz and approximately 12 GB RAM). The code is available on GitHub.\footnote{\url{https://github.com/afhamash/optimal-average-fidelity}}

\subsection{Comparing performance of SDP and Omega FP algorithm for optimal average fidelity}

The semidefinite program as defined in Definition \ref{def:SDP} can be solved numerically using standard convex optimization libraries. However, owing to the size of the matrices over which we optimize --- positive semidefinite matrices of size $(n+1)d \times (n+1)d$ --- the process is quite intractable for even moderately large $n$ and $d$. This intractability can be partially resolved by defining an \emph{alternate SDP} for optimal average fidelity which reformulates the problem into solving $n$ SDPs over matrices of size $2d \times 2d$ (see Appendix~\ref{app:AltSDP}). The Alternate SDP provides a more favorable scaling as the number of states $n$ in the ensemble increases. Moreover, the runtime scales (roughly) linear in the number of states as compared to superlinear for the original SDP. 

As seen in Fig. \ref{Fig:SDPvFP}, the $\Omega$ fixed point method vastly outperforms both SDPs in terms of runtime. Though both the alternate SDP and FP method scale linearly in $n$, The FP method is orders of magnitude faster. 
In particular, for 5 qubit states ($d = 32$), the FP method was faster than alternate SDP by a factor of 68 on average. The plots also show that the time taken for the FP method to obtain a solution in the 5 qubit scenario is comparable to the time it takes the alternate SDP to solve the 1 qubit case. The figure also shows how intractable solving the original SDP can be. For even just 1 qubit ($d=2$) and $n=20$ states, it takes more time than the FP method takes for 5 qubits and $n = 20$. Moreover, the FP method can be easily parallelized at each iteration, as the $n$ different terms~(see Eq.~\eqref{eq:Omega}) in the sum can be computed in parallel, thereby further boosting performance.

\subsection{Simulating Bayesian tomography}

\begin{figure*}[t]
    \center
    \includegraphics[scale = 0.535]{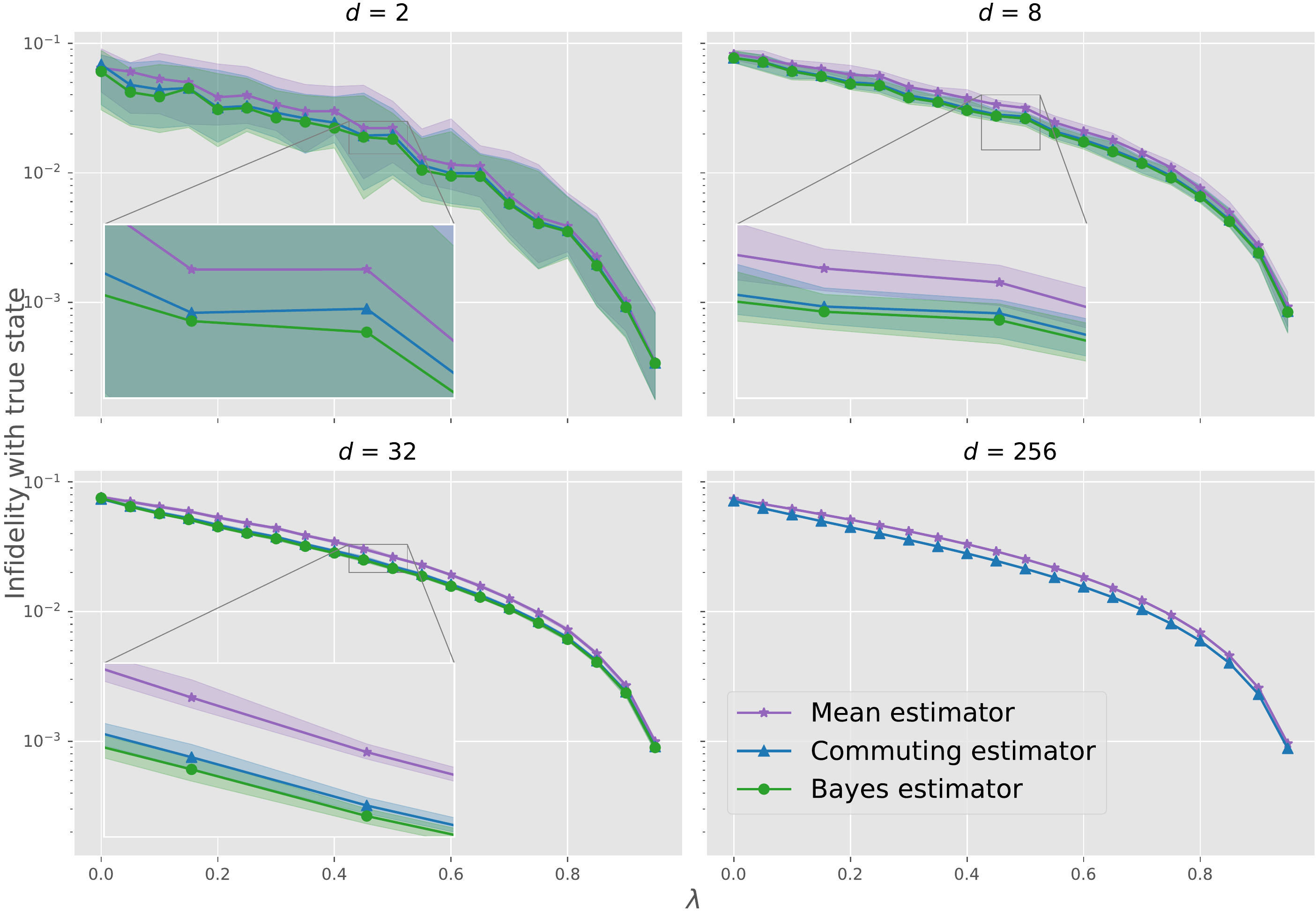}
    \caption{Infidelity with true state as a function of $\lambda$ for various dimensions ($n = 20$) plotted in log scale. $\lambda \in [0,1]$ is a parameter to simulate measurement with $\lambda = 0$ indicating 0 measurements and $\lambda = 1$ indicating infinitely many measurements. Each point is the median of 50 iterations. Interquartile regions are shaded. As seen in the inset plots, the Commuting estimator can serve as a high quality proxy (better than the Mean estimator) to the Bayes estimator while being inexpensive to compute.} \label{Fig:LambdaTomo}
\end{figure*}  

In the second set of experiments, we simulate Bayesian tomography. The results are presented in Fig.~\ref{Fig:LambdaTomo}. Since simulating measurement and computing the posterior distribution is expensive and outside the scope of this paper, we assign posterior weights as follows. We begin with a randomly generated \emph{true state} $\rho_\text{T}$. In Bayesian state estimation, we begin with a set of hypothesis states and an associated distribution. The posterior distribution would be peaked near the true state and as we increase the number of measurements, the sharpness of the peak increases. 
To simulate this, we introduce a parameter $\lambda \in [0,1]$, and then for randomly generated $\{\rho_i'\}_{i=1}^n$, we choose our particles as $\rho_i = \lambda \rho_\text{T} +(1 - \lambda) \rho_i'$. As $\lambda \to 1$, the particles $\{\rho_i'\}_{i=1}^n$ are closer to the true state $\rho_\text{T}$. The unnormalized weights are then chosen as $\operatorname{F}(\rho_i, \rho_\text{T})$, which is then normalized to serve as the weights $p_i$:
\begin{equation}
    p_i \propto \operatorname{F}(\rho_i, \rho_\text{T}), \quad \sum_{i=1}^{n} p_i = 1. 
\end{equation}

This allows assignment of higher weights to $\rho_i$s that are closer to the true state $\rho_\text{T}$. We then compute the Mean estimator $\sigma_\text{M} = \sum_{i=1}^{n} p_i\rho_i$, the Commuting estimator $\sigma' = \Gamma\left(\left(\sum_{i=1}^{n} p_i \rho_i^{1/2}\right)^2\right)$ and the Bayes estimator (optimal estimator) $\sigma_\sharp$. In tomography, one is usually interested in the behaviour of infidelity with the true state as a function of the number of measurements. To simulate this, we vary $\lambda$ from 0 to 1, as $\lambda = 0$ would correspond to the hypothesis states being randomly initialized states (zero measurements) and $\lambda = 1$ would correspond to all the particles being equal to the true state (infinitely many measurements). Fig.~\ref{Fig:LambdaTomo} plots the infidelity of the Bayes estimator, Commuting estimator, and Mean estimator with the true state for various dimensions and $n = 20$. 

For higher dimensions, we drop the Bayes estimator due to the computational costs while noting that the Commuting estimator $\sigma'$ remains a good alternative while being inexpensive to compute.

\subsection{Tightness of upper bounds and other estimators with optimum}

\begin{figure*}
    \center
    \includegraphics[scale = 0.495]{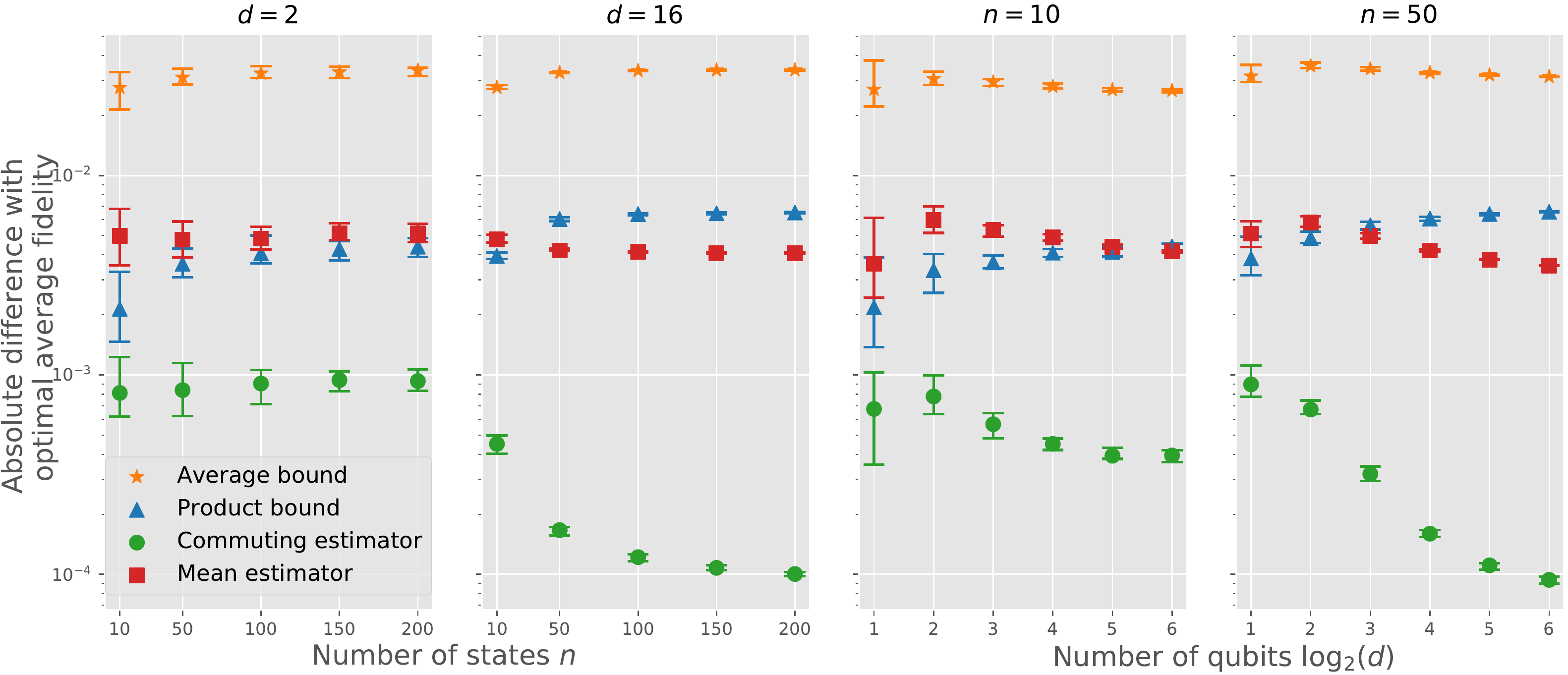}
    \caption{Tightness of quantities of interest with optimal average fidelity as a function of (a) number of states $n$ and (b) number of qubits ($\log_2(d)$). We plot the \emph{absolute difference} with the optimal average fidelity of (i) Average bound~\eqref{eq:averagebound}, (ii) Product bound~\eqref{eq:ProdBounds}, (iii)~average fidelity of Commuting estimator~\eqref{eq:SigmaComm}, and (iv)~average fidelity of Mean estimator $\sigma_{\text{M}} = \sum_{i=1}^{n} p_i \rho_i$. Note that Average bound and Product bound are upper bounds and average fidelities of Mean estimator and Commuting estimator are lower bounds. Each data point is the median of 50 iterations and ticks correspond to interquartile regions. See Eq. \eqref{eq:PlotBounds} for details on the quantities being plotted.} \label{Fig:Bounds}
\end{figure*} 

In Fig.~\ref{Fig:Bounds}, we numerically demonstrate the tightness of the two upper bounds and the fidelity achieved by the Commuting estimator and Mean estimator with the maximum average fidelity for randomly generated ensembles for various dimensions and number of states. More formally, for randomly generated ensembles $(\mathsf R, p)$, we plot the quantity |$f(\sigma_\sharp) - g$|, where
\begin{equation} \label{eq:PlotBounds}
    g = 
    \begin{cases}
     \sqrt{f(\sigma_\text{M})} & \text{(Average bound)}, \\
     \sqrt{\sum_{i=1}^{n} p_i p_j \operatorname{F}(\rho_i,\rho_j)} & (\text{Product bound}), \\
     f(\sigma') & (\text{Commuting estimator}), \\
     f(\sigma_\text{M}) & (\text{Mean estimator}). \\
    \end{cases}   
\end{equation}
Note that the first two are upper bounds while the last two are lower bounds on optimal average fidelity $f(\sigma_\sharp)$. Since the average fidelity of the Mean estimator and Product bounds are different kinds of bounds (lower and upper respectively), their crossing is not unexpected, and it simply means that the lower bound gets closer to the optimum than the upper bound. As the plots show, the Product bound and average fidelity of the Commuting estimator are quite close to the optimal average fidelity.

\section{Conclusion} \label{Sec:Conclusion}
In this work, we present algorithms for identifying states that maximize average fidelity over arbitrary finite ensembles of quantum states. We have constructed semidefinite programs which solve this problem, from which we derive faster fixed-point algorithms for the scenario where all the states in the ensemble are full rank. The fixed point methods are orders of magnitude faster than the semidefinite programs. We also derive heuristic approximations for the optimal state which are exact when the states in the ensemble commute pairwise. Furthermore, we derive novel upper and lower bounds for maximum average fidelity achievable by any quantum state, which are saturated when all the states commute pairwise. Finally, we present numerical experiments to complement our theoretical findings. These results solve open problems in Bayesian quantum tomography and are of independent theoretical interest.

The general technique employed in this work --- constructing a semidefinite program and deriving a fixed-point equation from its complementary slackness relations --- may not directly translate to other figures of merit like square fidelity and trace distance. However, it would be valuable to study how to extend these ideas to maximize averages of other popular figures of merit such as square fidelity and trace distance. 

\section{Acknowledgements}

We thank Afrad Basheer and Daniel Burgarth for their helpful discussions. This work was supported by the Australian Department of Industry, Innovation and Science (Grant No. AUSMURI000002). AA was supported by Sydney Quantum Academy.

\printbibliography

@article{granade_qinfer_2017,
	title = {{QInfer}: {Statistical} inference software for quantum applications},
	volume = {1},
	shorttitle = {{QInfer}},
	doi = {10.22331/q-2017-04-25-5},
	language = {en-GB},
	journal = {Quantum},
	author = {Granade, Christopher and Ferrie, Christopher and Hincks, Ian and Casagrande, Steven and Alexander, Thomas and Gross, Jonathan and Kononenko, Michal and Sanders, Yuval},
	month = {4},
	year = {2017},
	note = {Publisher: Verein zur Förderung des Open Access Publizierens in den Quantenwissenschaften},
	pages = {5}
}

@article{Kueng2015,
  title={Near-optimal quantum tomography: estimators and bounds},
  author={Kueng, Richard and Ferrie, Christopher},
  journal={New Journal of Physics},
  volume={17},
  number={12},
  pages={123013},
  year={2015},
  publisher={IOP Publishing},
  doi={10.1088/1367-2630/17/12/123013}
%   \href{https://doi.org/10.1088/1367-2630/17/12/123013}
}

@article{Bhatia2019,
  title={On the Bures--Wasserstein distance between positive definite matrices},
  author={Bhatia, Rajendra and Jain, Tanvi and Lim, Yongdo},
  journal={Expositiones Mathematicae},
  volume={37},
  number={2},
  pages={165--191},
  year={2019},
  publisher={Elsevier},
  doi = {10.1016/j.exmath.2018.01.002}
}

@misc{Ferrie2016,
      title={Bayes estimator for multinomial parameters and Bhattacharyya distances}, 
      author={Christopher Ferrie and Robin Blume-Kohout},
      year={2016},
      eprint={1612.07946},
      archivePrefix={arXiv},
      primaryClass={math.ST}
}

@book{Watrous2018, 
      place={Cambridge}, 
      title={The Theory of Quantum Information}, 
      DOI={10.1017/9781316848142}, 
      publisher={Cambridge University Press}, 
      author={Watrous, John}, 
      year={2018}
}

@book{Hornmatrix, 
      place={Cambridge}, title={Matrix Analysis}, DOI={10.1017/CBO9780511810817}, publisher={Cambridge University Press}, author={Horn, Roger A. and Johnson, Charles R.}, year={1985}}

@book{NielsenChuang, place={Cambridge}, title={Quantum Computation and Quantum Information: 10th Anniversary Edition}, DOI={10.1017/CBO9780511976667}, publisher={Cambridge University Press}, author={Nielsen, Michael A. and Chuang, Isaac L.}, year={2010}}

@article{Hradil97,
  title = {Quantum-state estimation},
  author = {Hradil, Z.},
  journal = {Phys. Rev. A},
  volume = {55},
  issue = {3},
  pages = {R1561--R1564},
  numpages = {0},
  year = {1997},
  month = {3},
  publisher = {American Physical Society},
  doi = {10.1103/PhysRevA.55.R1561}
}

@book{Wilde2013, 
    place={Cambridge}, 
    title={Quantum Information Theory}, 
    DOI={10.1017/CBO9781139525343}, 
    publisher={Cambridge University Press}, 
    author={Wilde, Mark M.}, 
    year={2013}}

@article{Bagan2006,
  title = {Optimal full estimation of qubit mixed states},
  author = {Bagan, E. and Ballester, M. A. and Gill, R. D. and Monras, A. and Mu\~noz-Tapia, R.},
  journal = {Phys. Rev. A},
  volume = {73},
  issue = {3},
  pages = {032301},
  numpages = {18},
  year = {2006},
  month = {03},
  publisher = {American Physical Society},
  doi = {10.1103/PhysRevA.73.032301}
}

@article{Blume2010BME,
	doi = {10.1088/1367-2630/12/4/043034},
	year = 2010,
	month = {04},
	publisher = {{IOP} Publishing},
	volume = {12},
	number = {4},
	author = {Robin Blume-Kohout},
	title = {Optimal, reliable estimation of quantum states},
	journal = {New Journal of Physics},
}

@article{Blume2010HMLE,
  title = {Hedged Maximum Likelihood Quantum State Estimation},
  author = {Blume-Kohout, Robin},
  journal = {Phys. Rev. Lett.},
  volume = {105},
  issue = {20},
  numpages = {4},
  year = {2010},
  month = {11},
  publisher = {American Physical Society},
  doi = {10.1103/PhysRevLett.105.200504}
}

@book{
    Boyd2004, 
    place={Cambridge}, 
    title={Convex Optimization}, 
    DOI={10.1017/CBO9780511804441}, 
    publisher={Cambridge University Press}, 
    author={Boyd, Stephen and Vandenberghe, Lieven}, 
    year={2004}
}

@book{Barvinok2002,
  title={A course in convexity},
  author={Barvinok, Alexander},
  volume={54},
  year={2002},
  publisher={American Mathematical Soc.},
  DOI = {10.1090/gsm/054}
}

@incollection{BhatiaPD,
  title={Positive definite matrices},
  author={Bhatia, Rajendra},
  booktitle={Positive Definite Matrices},
  year={2009},
  DOI={10.1515/9781400827787},
  publisher={Princeton university press}
}

@article{Gross2010,
  title = {Quantum State Tomography via Compressed Sensing},
  author = {Gross, David and Liu, Yi-Kai and Flammia, Steven T. and Becker, Stephen and Eisert, Jens},
  journal = {Phys. Rev. Lett.},
  volume = {105},
  issue = {15},
  pages = {150401},
  numpages = {4},
  year = {2010},
  month = {10},
  publisher = {American Physical Society},
  doi = {10.1103/PhysRevLett.105.150401}
}

@article{Flammia2012,
	doi = {10.1088/1367-2630/14/9/095022},
	year = 2012,
	month = {9},
	publisher = {{IOP} Publishing},
	volume = {14},
	number = {9},
	pages = {095022},
	author = {Steven T Flammia and David Gross and Yi-Kai Liu and Jens Eisert},
	title = {Quantum tomography via compressed sensing: error bounds, sample complexity and efficient estimators},
	journal = {New Journal of Physics}
}

@article{Huszar2012,
  title = {Adaptive Bayesian quantum tomography},
  author = {Husz\'ar, F. and Houlsby, N. M. T.},
  journal = {Phys. Rev. A},
  volume = {85},
  issue = {5},
  pages = {052120},
  numpages = {5},
  year = {2012},
  month = {5},
  publisher = {American Physical Society},
  doi = {10.1103/PhysRevA.85.052120},
}

@article{Granade2016,
	doi = {10.1088/1367-2630/18/3/033024},
  	year = 2016,
	month = {3},
  	publisher = {{IOP} Publishing},
  	volume = {18},
  	number = {3},
  	pages = {033024},
  	author = {Christopher Granade and Joshua Combes and D G Cory},
  	title = {Practical Bayesian tomography},
  	journal = {New Journal of Physics}
}

@article{Kliesch2016,
	doi = {10.1109/tit.2016.2606500},
  	year = 2016,
	month = {12},
  	publisher = {Institute of Electrical and Electronics Engineers ({IEEE})},
  	volume = {62},
  	number = {12},
	pages = {7445--7463},
  	author = {Martin Kliesch and Richard Kueng and Jens Eisert and David Gross},
    title = {Improving Compressed Sensing With the Diamond Norm},
    journal = {{IEEE} Transactions on Information Theory}
}

@article{Rehacek2001,
  title = {Iterative algorithm for reconstruction of entangled states},
  author = {Rehácek, J. and Hradil, Z. and Ježek, M.},
  journal = {Phys. Rev. A},
  volume = {63},
  issue = {4},
  pages = {040303},
  numpages = {4},
  year = {2001},
  month = {3},
  publisher = {American Physical Society},
  doi = {10.1103/PhysRevA.63.040303}
}

@article{Guta2020,
	doi = {10.1088/1751-8121/ab8111},
	year = 2020,
	month = {4},
	publisher = {{IOP} Publishing},
	volume = {53},
	number = {20},
	pages = {204001},
	author = {M Gu{\c{t}}{\u{a}} and J Kahn and R Kueng and J A Tropp}
}

@article{DAriano2007,
  title = {Optimal Data Processing for Quantum Measurements},
  author = {D'Ariano, G. M. and Perinotti, P.},
  journal = {Phys. Rev. Lett.},
  volume = {98},
  issue = {2},
  pages = {020403},
  numpages = {4},
  year = {2007},
  month = {1},
  publisher = {American Physical Society},
  doi = {10.1103/PhysRevLett.98.020403}
}

@article{CVXPY1,
  author  = {Steven Diamond and Stephen Boyd},
  title   = {{CVXPY}: {A} {P}ython-embedded modeling language for convex optimization},
  journal = {Journal of Machine Learning Research},
  year    = {2016},
  volume  = {17},
  number  = {83},
  pages   = {1--5},
  doi = {10.48550/arxiv.1603.00943}
}

@article{CVXPY2,
author = {Akshay Agrawal and Robin Verschueren and Steven Diamond and Stephen Boyd},
title = {A rewriting system for convex optimization problems},
journal = {Journal of Control and Decision},
volume = {5},
number = {1},
year  = {2018},
publisher = {Taylor & Francis},
doi = {10.1080/23307706.2017.1397554}
}

@article{You2021,
         title = {Minimizing Quantum Renyi Divergences via Mirror Descent with Polyak Step Size},
         author = {Jun-Kai You and Hao-Chung Cheng and Yen-Huan Li },
         year = {2021},
         eprint = {2109.06054v2},
         archivePrefix = {arXiv},
         primaryClass ={cs.IT}
        }

\clearpage

\appendix

\section{Alternate semidefinite program for optimal fidelity} \label{app:AltSDP}
    A more numerically tractable SDP (which is still not as tractable as the FP algorithm) to solve the maximization problem \eqref{eq:maximprob} can be constructed as follows. Let $\{\rho_i\}_{i=1}^n \subset \operatorname{D}(\mathcal H)$  be a collection of states. The alternate SDP for optimal fidelity is formulated as $n$ different Watrous SDPs (see Eq.\eqref{SDP:FidSDP}), with the constraint that the matrix variable $\sigma$ is the same in each of the $n$ SDPs. That is,

\begin{equation} \label{eq:AltSDP}
    \begin{aligned}
    &\text{Primal problem}\\
    \text{maximize : } & \frac12 \sum_{i=1}^{n} p_i \operatorname{Tr}(X_i +X_i^*) = \sum_{i=1}^{n} p_i\Re(\operatorname{Tr}(X_i))\\
    \text{subject to : } & \begin{pmatrix}    \rho_i & X_i \\
        X_i^* & \sigma\end{pmatrix} \geq 0, \quad \text{ for all } i \in \{1,\ldots,n\}.  
    \end{aligned}
\end{equation}

This SDP also achieves the primal optimum $\max _{\sigma \in \operatorname{D}(\mathcal H) }f(\sigma)$ while being more numerically tractable. We compare the performance of this SDP along with the original SDP and the FP algorithm in Fig.~\ref{Fig:SDPvFP}.

\section{Runtime comparison of fixed point algorithms} \label{app:AddNumerics}
\begin{figure*}[t]
    \center
    \includegraphics[scale = 0.55]{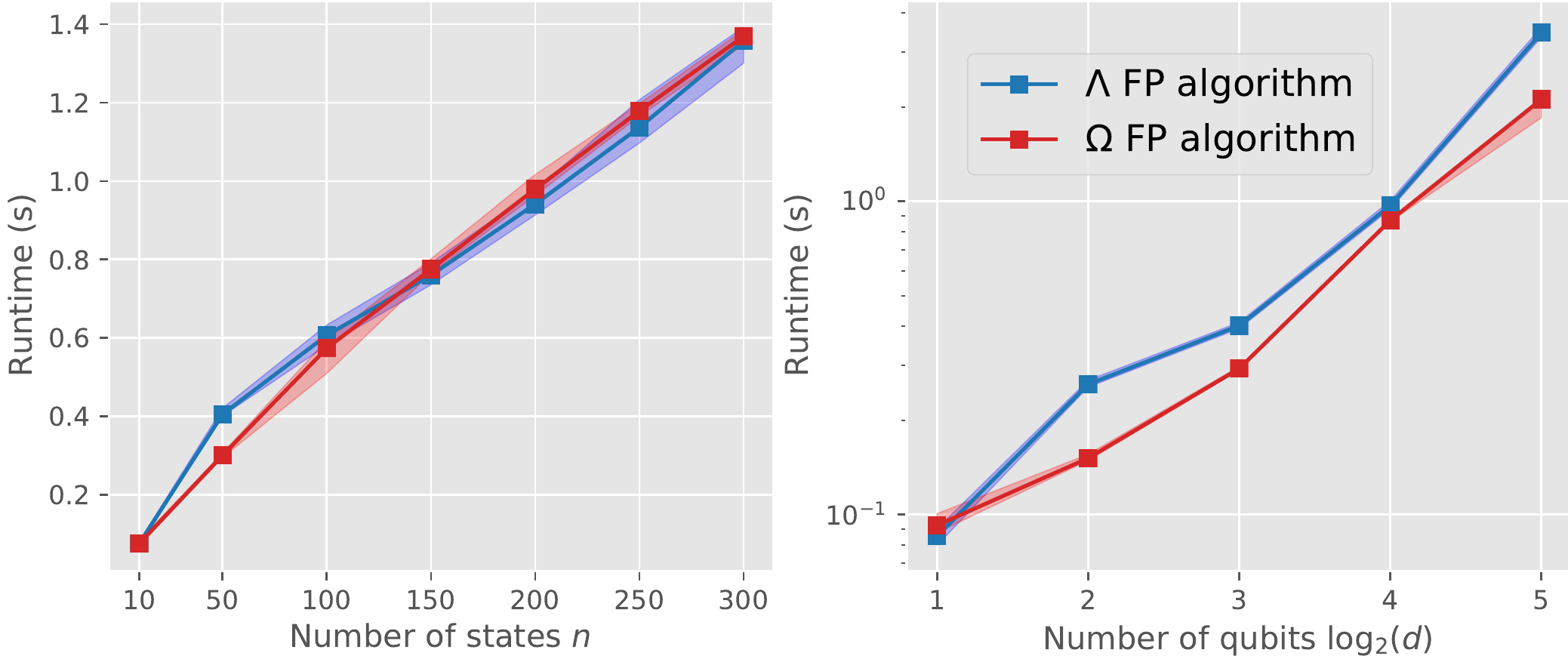}
    \caption{Runtime comparison of $\Lambda$ FP algorithm~\eqref{eq:Lambda} and $\Omega$ FP algorithm~\eqref{eq:Omega} as a function of (a) Number of states for $d = 8$ and (b) Number of qubits ($\log_2(d)$) for $n = 50$. The $\Omega$ FP algorithm offers a superior scaling of runtime as compared to $\Lambda$ FP algorithm. Each data point is the median of 50 runs and interquartile regions are shaded.} \label{Fig:FPCompare}
\end{figure*}

When all the states in the ensemble are full rank, numerically both the FP algorithms are seen to converge to the optimal state for any starting point. We now turn to numerics to compare the performance (runtime) of the two fixed point methods. As seen in Fig.~\ref{Fig:FPCompare}, the runtime performance of $\Omega$ fixed point algorithm~\eqref{eq:Omega} and $\Lambda$ fixed point algorithm~\eqref{eq:Lambda} are comparable. Here we choose stopping tolerance $\epsilon = 10^{-5}$. Since the convergence of the $\Omega$ fixed-point algorithm is theoretically guaranteed, it should be preferred over $\Lambda$ fixed-point algorithm even though the former has a more complicated form than the latter.

\end{document}